\documentclass[aps,pra,11pt,onecolumn,superscriptaddress,floatfix,nofootinbib,showpacs,longbibliography]{revtex4-2}
\usepackage{comment}
\usepackage{amsmath,amssymb,amsthm,bm,amsfonts,mathrsfs,bbm}
\usepackage[colorlinks=true, citecolor=blue, urlcolor=blue]{hyperref}
\usepackage[caption = false]{subfig}
\usepackage{times}
\usepackage{xcolor,colortbl}
\usepackage{ragged2e}
\usepackage{tikz}
\usepackage{scalerel}
\usetikzlibrary{calc}
\usetikzlibrary {arrows.meta}
\usepackage{leftindex}
\usepackage{braket}

\usepackage[utf8]{inputenc}  
\usepackage[T1]{fontenc}     
\usepackage[british]{babel}  
\usepackage[sc,osf]{mathpazo}
\usepackage[scaled=0.86]{berasans}  
\usepackage[colorlinks=true, citecolor=blue, urlcolor=blue]{hyperref}
\makeatletter
\newcommand{\setword}[2]{%
  \phantomsection
  #1\def\@currentlabel{\unexpanded{#1}}\label{#2}%
}
\makeatother
\usepackage{graphicx} 
\usepackage[babel]{microtype}
\usepackage{blkarray}
\usepackage{amsmath}

\usepackage{amsmath,amssymb,amsthm,bm,amsfonts,mathrsfs,bbm} 
\usepackage{xspace}  
\usepackage{pgf,tikz}
\usepackage{xcolor}
\usepackage{multirow}
\usepackage{array}
\usepackage{bigstrut}
\usepackage{braket}
\usepackage{color}
\usepackage{natbib}

\usepackage{multirow}
\usepackage{mathtools}
\usepackage{float}
\usepackage[caption = false]{subfig}
\usepackage{xcolor,colortbl}
\usepackage{color}

\newcommand{\be}{\begin{equation}}
\newcommand{\ee}{\end{equation}}
\newcommand{\ba}{\begin{eqnarray}}
\newcommand{\ea}{\end{eqnarray}}

\newtheorem{theorem}{Theorem}

\newtheorem{proposition}{Proposition}

\usepackage{diagbox}
\usepackage{multirow}
\usepackage{tabularx}
\usepackage{comment}

\usepackage{slashbox}
\usepackage{tikz}
\usetikzlibrary{decorations.markings}
\usepackage{pgfplots}
\usepackage{verbatim}
\usepgfplotslibrary{fillbetween}
\usetikzlibrary{patterns.meta}
\usetikzlibrary{patterns}
\usetikzlibrary{arrows}

\usepackage{blindtext}
\usepackage{tcolorbox}

\def\MarkLt{4pt}
\def\MarkSep{2pt}

\tikzset{
  TwoMarks/.style={
    postaction={decorate,
      decoration={
        markings,
        mark=at position #1 with
          {
              \begin{scope}[xslant=0.2]
              \draw[line width=\MarkSep,white,-] (0pt,-\MarkLt) -- (0pt,\MarkLt) ;
              \draw[-] (-0.5*\MarkSep,-\MarkLt) -- (-0.5*\MarkSep,\MarkLt) ;
              \draw[-] (0.5*\MarkSep,-\MarkLt) -- (0.5*\MarkSep,\MarkLt) ;
              \end{scope}
          }
       }
    }
  },
  TwoMarks/.default={0.5},
}

\def\>{\rangle}
\def\<{\langle}

\usepackage{stmaryrd}

\usepackage{centernot}
\usepackage{subfig}

\usepackage{diagbox}
\usepackage{multirow}
\usepackage{tabularx}
\usepackage{circuitikz}

\begin{document}

\title{Nonlocality-Assisted Enhancement of Error-Free Communication \\in Noisy Classical Channels}

\author{Kunika Agarwal}
\affiliation{Department of Physics of Complex Systems, S. N. Bose National Center for Basic Sciences, Block JD, Sector III, Salt Lake, Kolkata 700106, India.}

\author{Sahil Gopalkrishna Naik}
\affiliation{Department of Physics of Complex Systems, S. N. Bose National Center for Basic Sciences, Block JD, Sector III, Salt Lake, Kolkata 700106, India.}

\author{Ananya Chakraborty}
\affiliation{Department of Physics of Complex Systems, S. N. Bose National Center for Basic Sciences, Block JD, Sector III, Salt Lake, Kolkata 700106, India.}

\author{Samrat Sen}
\affiliation{Department of Physics of Complex Systems, S. N. Bose National Center for Basic Sciences, Block JD, Sector III, Salt Lake, Kolkata 700106, India.}

\author{Pratik Ghosal}
\affiliation{Department of Physics of Complex Systems, S. N. Bose National Center for Basic Sciences, Block JD, Sector III, Salt Lake, Kolkata 700106, India.}

\author{Biswajit Paul}
\affiliation{Department Of Mathematics, Balagarh Bijoy Krishna Mahavidyalaya, Balagarh, Hooghly-712501, West Bengal, India.}

\author{Manik Banik}
\affiliation{Department of Physics of Complex Systems, S. N. Bose National Center for Basic Sciences, Block JD, Sector III, Salt Lake, Kolkata 700106, India.}

\author{Ram Krishna Patra}
\affiliation{Department of Physics of Complex Systems, S. N. Bose National Center for Basic Sciences, Block JD, Sector III, Salt Lake, Kolkata 700106, India.}

\begin{abstract}
\noindent The zero-error capacity of a noisy classical channel quantifies its ability to transmit information with absolute certainty, i.e., without any error. Unlike Shannon's standard channel capacity, which remains unaffected by pre-shared no-signaling correlations, zero-error capacity can be enhanced through nonlocal correlations. In this work, we first investigate zero-error communication utility of nonlocal correlations arising in the \(2\)-\(2\)-\(m\) Bell scenario, where two parties have two inputs and \(m\) possible outcomes per input. For all \(m \geq 2\), we construct examples of noisy classical channels with zero zero-error capacity that, when assisted by extremal \(2\)-\(2\)-\(m\) nonlocal correlations, can transmit one bit of information. While nonlocal correlations arising from quantum entangled states cannot achieve a positive zero-error capacity for these channels, they can significantly enhance the success probability of one-shot transmission of  classical messages. Extending this analysis to the \(2\)-\(m\)-\(2\) Bell scenario, we identify channels with zero zero-error capacity that can nonetheless perfectly transmit \(\log m\) bits of information when assisted by corresponding extremal nonlocal correlations. Our findings underscore the versatile utility of Bell nonlocal correlations in achieving zero-error communication.  
\end{abstract}

\maketitle
\section{Introduction}
\noindent The nonlocal nature of quantum mechanics, established by the seminal Bell’s theorem \cite{Bell1964}, marks a profound departure of quantum theory from the deeply ingrained classical notion of `local realism' \cite{Bell1966, Mermin1993, Brunner2014}. Empirical demonstrations of quantum nonlocality typically rely on violations of Bell-type inequalities, which analyze outcome statistics from incompatible local measurements on entangled systems \cite{Freedman1972, Aspect1981, Aspect1982(1), Aspect1982(2), Zukowski1993, Weihs1998, Hensen2015, BIG2018, Storz2023}. In the bipartite setting, the  Clauser-Horne-Shimony-Holt (CHSH) inequality remains one of the most widely studied tests, involving two binary measurements per party \cite{Clauser1969}. Subsequent works showed that nonlocality exhibited by quantum mechanics, as in the CHSH test, is limited compared to broader theories that satisfy the no-signaling (NS) principle \cite{Cirelson1980, Popescu1994, Popescu2014}. This insight has fostered a research area merging quantum information with fundamental principles to explore deeper aspects of physical theory \cite{Brassard2006, vanDam2012, Buhrman2010, Pawlowski2009, Navascus2009, Fritz2013, Banik2013,Naik2022,Patra2023,Naik2023,Ghosh2024}. Beyond two-input and two-output case, higher input-output nonlocal correlations reveal additional exotic behaviors (see \cite{Cope2019} and references therein), prompting critical inquiries into the mathematical structure underlying quantum mechanics \cite{Palazuelos2016}. Recent advances have resolved longstanding mathematical challenges \cite{Slofstra2019(1), Slofstra2019(2), Ji2021}, shedding new light on foundational aspects of physical theories \cite{Cabello2023}.\\

\noindent Beyond its foundational significance, philosophical appeal, and mathematical intricacy, study of Bell nonlocality also offers substantial practical value \cite{Brunner2014}. Nonlocal correlations have applications in secure cryptographic key generation, communication complexity, device-independent randomness certification, dimension witnessing, Bayesian game theory, and zero-error communication theory \cite{Ekert1991, Barrett2005, Pironio2010, Chaturvedi2015, Mukherjee2015, Scarani2012, Brunner2013, Pappa2015, Roy2016, Banik2019, Prevedel2011, Cubitt2010, Cubitt2011, Yadavalli2022, Alimuddin2023}. In this article we investigate how nonlocal correlations can enhance the zero-error communication capacity of noisy classical channels. Origin of zero-error information theory dates back to Shannon’s seminal 1956 work \cite{Shannon1956}. For a classical channel \(\mathcal{N}\), its single-shot zero-error capacity \(\mathrm{C}^\circ(\mathcal{N})\) quantifies the maximum rate of error-free message transmission in a single channel use. Notably, pre-shared entanglement and more general NS correlations can extend assisted zero-error capacities, \(\mathrm{C}^\circ_E\) and \(\mathrm{C}^\circ_{NS}\), beyond the unassisted zero-error capacity \(\mathrm{C}^\circ\), for certain noisy classical channels \cite{Cubitt2010, Cubitt2011, Yadavalli2022}. Here, we start by examining the \(2\)-\(2\)-\(m\) Bell scenario, where Alice and Bob each have two inputs, yielding \(m\) possible outcomes per input. For each \(m \geq 2\), we construct a classical channel \(\mathcal{N}_m\) with zero unassisted zero-error capacity that can still transmit one bit of information without any error when assisted by extremal \(2\)-\(2\)-\(m\) nonlocal correlations. While the nonlocal correlation obtained from quantum entangled states cannot make the zero-error capacity of a classical channel positive if the unassisted capacity is zero \cite{Cubitt2011}, we demonstrate that entanglement can however enhance the probability of successful transmission of one-bit message in a single usage of these channels. We further explore the \(2\)-\(m\)-\(2\) Bell scenario, in which each party has \(m\) inputs and two possible outputs per input, identifying noisy classical channels with zero-error capacity that nonetheless allows perfect $\log m$ bits communication when assisted by extremal nonlocal correlations of \(2\)-\(m\)-\(2\) settings.

\section{Correlation experiment}
\noindent A Bell-type correlation experiment involves multiple causally separated parties performing local measurements on respective subsystems of a composite quantum system. In the bipartite case, let \( x \in \mathcal{X} \) and \( y \in \mathcal{Y} \) represent the measurement choices of Alice and Bob, with outcomes \( a \in \mathcal{A} \) and \( b \in \mathcal{B} \), respectively. Repeated experiments yield a correlation or 'behavior' denoted \( P \equiv \{p(a,b|x,y)\} \), where \( p(a,b|x,y) \) is the probability of outcomes \( a \) and \( b \) given inputs \( x \) and \( y \). A correlation is called classical (or `local realistic') if it admits a decomposition \( p(a,b|x,y) = \int d\lambda \, p(\lambda) p(a|x,\lambda) p(b|y,\lambda) \), where \( \lambda \) is a shared classical random variable with distribution \( p(\lambda) \), independent of \( x \) and \( y \) \cite{Bell1964, Bell1966, Mermin1993, Brunner2014}. Quantum correlations, on the other hand, allow for representations of the form \( p(a,b|x,y) = \operatorname{Tr}[(\pi^a_x \otimes \pi^b_y) \rho] \), where \( \rho \in \mathcal{D}(\mathcal{H}_A \otimes \mathcal{H}_B) \) is a density matrix on the joint Hilbert space \( \mathcal{H}_A \otimes \mathcal{H}_B \), and \( \pi^a_x \in \mathcal{P}(\mathcal{H}_A) \) and \( \pi^b_y \in \mathcal{P}(\mathcal{H}_B) \) are positive operators satisfying \( \sum_a \pi^a_x = \mathbf{I}_A~\forall~x \) and \( \sum_b \pi^b_y = \mathbf{I}_B~\forall~y \), with \( \mathbf{I} \) denoting the identity operator. For additional characterizations of quantum correlations, see \cite{Slofstra2019(1), Slofstra2019(2), Ji2021}. A no-signaling (NS) correlation requires only that Alice’s marginal probabilities remain independent of Bob’s input and vice versa, i.e., \( p(a|x,y) = p(a|x,y') \) and \( p(b|x,y) = p(b|x',y) \). The sets of local correlations (\( \mathcal{L} \)) and NS correlations (\( \mathcal{NS} \)) form convex polytopes in some \( \mathbb{R}^n \), while the quantum correlations (\( \mathcal{Q} \)) form a convex set strictly between them, i.e., \( \mathcal{L} \subsetneq \mathcal{Q} \subsetneq \mathcal{NS} \). Both \( \mathcal{L} \) and \( \mathcal{NS} \) are characterized by extreme points, while \( \mathcal{Q} \) is generally not topologically closed \cite{Slofstra2019(2)}. \\

\noindent The \( 2 \)-\( 2 \)-\( m \) Bell scenario, where \( |\mathcal{X}| = |\mathcal{Y}| = 2 \) and \( |\mathcal{A}| = |\mathcal{B}| = m \geq 2 \), considers inputs \( x, y \in \{0,1\} \) and outputs \( a, b \in \{0,1,\dots,m-1\} \). The nonlocal extreme points \( P^{(2,m)}_{k} \equiv \{p_k(a,b|x,y)\} \) of the \( \mathcal{NS} \) set, up to relabeling of inputs and outputs, are given by \cite{Barrett2005(1)}:
\begin{align}
p_k(a,b|x,y) \equiv \begin{cases}
\frac{1}{k} & \text{if } (b - a) \mod k = x \cdot y, \\
0 & \text{otherwise},
\end{cases}
\end{align}
where \( a,b \in \{0,1,\ldots,k-1\}~\&~k \in \{2, \dots, m\} \) and superscript \( (u, v) \) denotes the configuration \( |\mathcal{X}| = |\mathcal{Y}| = u \) and \( |\mathcal{A}| = |\mathcal{B}| = v \). For \( k < m \), these correlations can be derived from sub-scenarios of the \( 2 \)-\( 2 \)-\( m \) Bell setting, where certain outputs are excluded. In the following, we focus on the case \( k = m \), representing the full-output extremal correlations, denoted simply as \( P_m \).

\section{Zero-error communication}
\noindent Given inputs \( i \in \mathcal{I} \) yielding outputs \( o \in \mathcal{O} \), a classical channel \(\mathcal{N}\) is defined by a stochastic matrix \(\mathbb{N} = \{p(o|i)~|~p(o|i) \geq 0, \, \sum_o p(o|i) = 1\}\). Shannon’s seminal channel coding theorem quantifies the channel's asymptotic efficiency through its capacity, defined as the maximum input-output mutual information across input distributions \cite{Shannon1948}. The reverse Shannon theorem confirms that pre-shared NS correlations, whether classical or beyond, do not increase this capacity \cite{Bennett2002, Winter2002, Bennett2014}. However, in many applications where error-free transmission is essential, and channel usage is limited, the relevant measure of channel utility shifts to the zero-error capacity, introduced by Shannon in a later seminal work \cite{Shannon1956}. Computing the zero-error capacity of classical channels involves a combinatorial graph-theoretic approach \cite{Korner1998}.\\

\noindent For a channel \(\mathcal{N}\), we denote its confusability graph \(\mathbb{G}(\mathcal{N})\equiv (\mathbb{V}, \mathbb{E})\), where vertices in \(\mathbb{V}\) represent inputs and edges in \(\mathbb{E}\) connect inputs that share at least one common output. The independence number \(\alpha(\mathbb{G})\) of a graph \(\mathbb{G}\) is the largest set of vertices that are pairwise non-adjacent. Without shared correlations, the single-shot zero-error capacity, \(\mathrm{C}^\circ(\mathcal{N})\), of a channel \(\mathcal{N}\) is given by \(\mathrm{C}^\circ(\mathcal{N}) = \log \alpha(\mathbb{G}(\mathcal{N}))\) bits. For two channels, the single-shot zero-error capacity of their combined use depends on the independence number of the strong product of their respective confusability graphs \cite{Self1}. While shared classical correlation cannot increase zero-error, it can be enhanced when assisted by quantum entanglement or more general NS correlations zero-error capacity, respectively denoted by \(\mathrm{C}^\circ_E\) and \(\mathrm{C}^\circ_{NS}\), when assisted by quantum entanglement or more general NS correlations \cite{Cubitt2010, Cubitt2011}. Nevertheless, if \(\mathrm{C}^\circ(\mathcal{N}) = 0\), entanglement cannot make \(\mathrm{C}^\circ_E(\mathcal{N})\) positive, while \(\mathrm{C}^\circ_{NS}(\mathcal{N})\) may still be nonzero.

\section{Results}
\noindent We start by defining a classical channel \(\mathcal{N}_3\) with inputs and outputs respectively denoted by the tuples \((i_1, i_2)\in\{0,1\}\times\{0,1,2\}\) and \((o_1, o_2)\in\{1,2,3,4\}\times\{0,1,2\}\) (see Fig.\ref{fig1}).
\begin{figure}[h!]
\centering
\begin{tikzpicture}
\draw[fill=blue!10](-2,0)--(-2,2)--(-1,2)--(-1,0)--(-2,0);
\draw[thin,black](-3,1.5)--(-2,1.5);
\draw[thin,black](-3,.6)--(-2,.6);
\draw[thin,black](-1,1.5)--(0,1.5);
\draw[thin,black](-1,.6)--(0,.6);
\draw[blue, fill=black] (-3,1.5) circle (.1);
\draw[blue, fill=black] (-3,.6) circle (.1);
\draw[blue, fill=black] (0,1.5) circle (.1);
\draw[blue, fill=black] (0,.6) circle (.1);
\node[] at (-1.5,1){\large{$\mathcal{N}_3$}};
\node[] at (-4.3,1.5){$\{0,1\}\ni i_1$};
\node[] at (-4.15,.6){$\{0,1,2\}\ni i_2$};
\node[] at (1.3,1.5){$o_1\in\{1,2,3,4\}$};
\node[] at (1.5,.6){$o_2\in\{0,1,2\}$};
\end{tikzpicture}
\caption{Diagrammatical representation of the channel $\mathcal{N}_3$.}\label{fig1}
\vspace{-.5cm}
\end{figure}
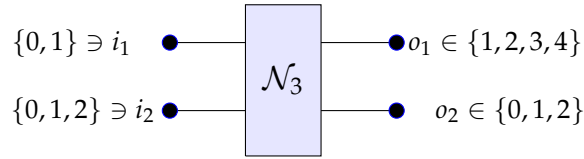\\
Outcome $o_2$ is conditioned on outcome $o_1$ as follows: \\
\begin{minipage}{.42\linewidth}
\begin{align*}
o_2:=\begin{cases}
i_1,~&\text{if } o_1 = 1;\\
i_2,~&\text{if } o_1 = 2;\\
i_1 \oplus_3 i_2,~&\text{if } o_1 = 3; \\
i_1 \oplus_3 \Pi(i_2),~&\text{if } o_1 = 4,
\end{cases}
\end{align*}
\end{minipage}
\begin{minipage}{.42\linewidth}
\begin{align}
~~~~\Pi(0):=0;\nonumber\\
~~~~\Pi(1):=2;\label{n3}\\
~~~~\Pi(2):=1.\nonumber
\end{align}
\end{minipage}\\\\
Here, \(\oplus_3\) denotes modulo $3$ addition.The channel exhibits two key symmetries: ($S_1$) each input tuple results in four different outputs, and ($S_2$) every pair of input tuples shares exactly one common output tuple. This implies that the confusability graph \(\mathbb{G}(\mathcal{N}_3)\) for \(\mathcal{N}_3\) is the complete graph \(K_6\). Accordingly, the unassisted zero-error capacity $\mathrm{C}^\circ(\mathcal{N}_3)$ as well as shared-randomness assisted zero-error capacity $\mathrm{C}^\circ_{SR}(\mathcal{N}_3)$ become zero. Our next theorem analyzes zero-error capacity of $\mathcal{N}_3$ when it is assisted with preshared $NS$ correlation.
\begin{table}[h!]
\begin{tabular}
{|c||c|c|c||c|c|c|}
\hline
$~~\text{out.} \backslash \text{in.}~~$ & ~~$(0,0)$~~ & ~~$(0,1)$~~ & ~~$(0,2)$~~ & ~~$(1,0)$~~ & ~~$(1,1)$~~ & ~~$(1,2)$~~ \\ \hline\hline
$\textcolor{purple}{(1,0)}$ & $\textcolor{purple}{1/4}$ & $\textcolor{purple}{1/4}$ & $\textcolor{purple}{1/4}$ & $0$ & $0$ &$ 0$ \\\hline
$\textcolor{blue}{(1,1)}$ & $0$ & $0$ & $0$ & $\textcolor{blue}{1/4}$ & $\textcolor{blue}{1/4}$ & $\textcolor{blue}{1/4}$ \\\hline
$(2,0)$ & $1/4$ & $0$ & $0$ & $1/4$ & $0$ & $0$ \\\hline
$(2,1)$ & $0$ & $1/4$ & $0$ & $0$ & $1/4$ & $0$ \\\hline
$(2,2)$ & $0$ & $0$ & $1/4$ & $0$ & $0$ & $1/4$ \\\hline
$(3,0)$ & $1/4$ & $0$ & $0$ & $0$ & $0$ & $1/4$ \\\hline
$(3,1)$ & $0$ & $1/4$ & $0$ & $1/4$ & $0$ & $0$ \\\hline
$(3,2)$ & $0$ & $0$ & $1/4$ & $0$ & $1/4$ & $0$ \\\hline
$(4,0)$ & $1/4$ & $0$ & $0$ & $0$ & $1/4$ & $0$ \\\hline
$(4,1)$ & $0$ & $0$ & $1/4$ & $1/4$ & $0$ & $0$ \\\hline
$(4,2)$ & $0$ & $1/4$ & $0$ & $0$ & $0$ & $1/4$ \\\hline
\end{tabular}
\caption{Stochastic matrix \(\mathbb{N}_3\equiv\{\text{Pr}(o_1, o_2 | i_1, i_2)\}\). The entries specify the outcome probabilities for each input. Here, inputs are deprecated in columns resulting addition of entries for each column to be $1$.}\label{tab1}
\end{table}
\begin{figure}[]
\centering
\begin{tikzpicture}
\draw[fill=blue!10](0,0)--(0,1.5)--(.75,1.5)--(.75,0)--(0,0);
\draw[thin,black](-1.5,1.2)--(0,1.2);
\draw[thin,black](-1.5,.3)--(0,0.3);
\draw[thin,black](.75,1.2)--(2.25,1.2);
\draw[thin,black](.75,.3)--(2.25,0.3);
\draw[blue, fill=black] (-1.5,1.2) circle (.08);
\draw[blue, fill=black] (-1.5,.3) circle (.08);
\draw[blue, fill=black] (2.25,1.2) circle (.08);
\draw[blue, fill=black] (2.25,0.3) circle (.08);
\node[] at (.4,.7){\large{$\mathcal{N}_3$}};
\node[] at (-1.6,1.5){$i_1=\gamma$};
\node[] at (-1.6,.6){$i_2=a$};
\node[] at (2.6,1.2){$o_1$};
\node[] at (2.6,.3){$o_2$};

\draw[fill=green!10](-1.75,-2)--(-1.75,-1)--(2.5,-1)--(2.5,-2)--(-1.75,-2);
\draw[thin,black](-1.25,-.2)--(-1.25,-1);
\draw[thin,black](-1.25,-2)--(-1.25,-2.8);
\draw[thin,black](2,-.2)--(2,-1);
\draw[thin,black](2,-2)--(2,-2.8);
\node[] at (-1.25,-.6)[rotate=180,inner sep=0pt] {\tikz\draw[-triangle 90](0,0) ;};
\node[] at (-1.25,-2.6)[rotate=180,inner sep=0pt] {\tikz\draw[-triangle 90](0,0) ;};
\node[] at (2,-.6)[rotate=180,inner sep=0pt] {\tikz\draw[-triangle 90](0,0) ;};
\node[] at (2,-2.6)[rotate=180,inner sep=0pt] {\tikz\draw[-triangle 90](0,0) ;};
\draw [thick,dotted] (-1.3,-2.8) .. controls (-2,-3) and (-3,-2) .. (-1.6,.18);

\node[] at (-2.26,-1.3)[rotate=-5,inner sep=0pt] {\tikz\draw[-triangle 90](0,0) ;};
\node[] at (.4,-1.5){\large{$P_3$}};
\node[] at (-.6,-.6){$x=\gamma$};
\node[] at (1.2,-.6){$y(o_1,o_2)$};
\node[] at (-.8,-2.4){$a$};
\node[] at (1.5,-2.4){$b$};

\draw[thin,black, dashed](-2.5,-3)--(-2.5,2)--(3.2,2)--(3.2,-3)--(-2.5,-3);
\draw[thin,black](-3,1.5)--(-2.5,1.5);
\draw[thin,black](3.2,-2.3)--(3.7,-2.3)--(3.7,-2.6);
\node[] at (-2.7,1.5)[rotate=270,inner sep=0pt] {\tikz\draw[-triangle 90](0,0) ;};
\node[] at (3.7,-2.7)[rotate=180,inner sep=0pt] {\tikz\draw[-triangle 90](0,0) ;};
\node[] at (-3.2,1.5){\textcolor[rgb]{0,0,1}{$\gamma$}};
\node[] at (4.1,-2.9){\textcolor[rgb]{0,.3,.8}{$\tilde{\gamma}(o_1,o_2,b)$}};
\node[] at (-3.,1.9){\textcolor[rgb]{0,0,1}{Alice}};
\node[] at (4,-3.2){\textcolor[rgb]{0,.3,.8}{Bob}};
\end{tikzpicture}    
\caption{Given the resources $\mathcal{N}_3~\&~P_3$, the protocol to transmit one-bit of message from Alice to Bob is depicted here.}\label{fig2}
\end{figure}
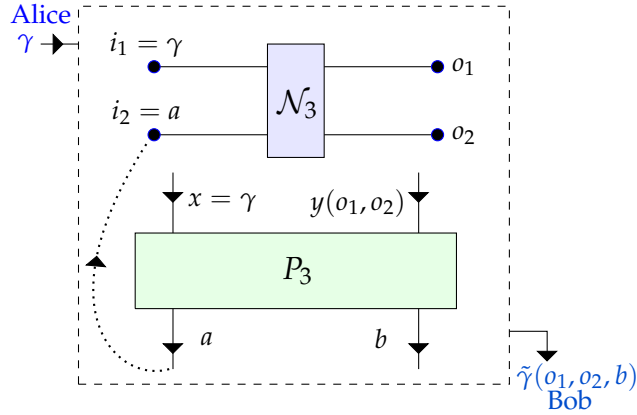
\begin{theorem}\label{theo1}
With a single use of the channel \(\mathcal{N}_3\), one bit of information can be perfectly transmitted from the sender to the receiver when assisted by the NS correlation \(P_3\), resulting in \(\mathrm{C}^\circ_{NS}(\mathcal{N}_3) \geq 1\).
\end{theorem}
\begin{proof}
Given a message \(\gamma \in \{0, 1\}\), Alice chooses her input \(x\) for the NS box \(P_3\) as \(\gamma\) and obtains the outcome \(a\). She then feeds the pair \((\gamma, a)\) as the input \((i_1, i_2)\) to the channel \(\mathcal{N}_3\). The receiver, Bob, selects his input \(y\) for his part of the NS box \(P_3\) based on the outcome \((o_1, o_2)\) received from \(\mathcal{N}_3\) and obtains an outcome \(b\). He then guesses \(\tilde{\gamma}\) depending on \((o_1, o_2)\) and \(b\). The protocol is illustrated in Fig. \ref{fig2}, and the calculation in Table \ref{tab2} demonstrates that Bob always guesses Alice’s message correctly. This completes the proof.
\end{proof}
\begin{table}[h!]
\begin{tabular}{|c|c|c|c|c|}
\hline
~output $o_1$~ & ~output $o_2$~~& ~input $y$~ &~outcome $b$~~& ~~~~guess $\tilde{\gamma}$~~~~\\\hline
$1$ & $\gamma$ & $\times$ & $\times$ & $o_2$ \\\hline
$2$ & $a$ & $1$ & $\gamma\oplus_3 a$ & $\Pi(o_{2})\oplus_3 b$\\\hline
$3$ & $\gamma\oplus_3 a$& $0$ & $a$& $o_2\oplus_3\Pi(b)$ \\\hline
$4$ & ~$\gamma\oplus_3 \Pi(a)$~ &$0$ &$a$ & $o_2\oplus_3 b $\\\hline
\end{tabular}
\caption{For all possible outputs $(o_1,o_2)$ from the channel $\mathcal{N}_3$ and the the outcomes $y$ of the NS box, Bob's guess $\tilde{\gamma}$ exactly matches with Alice's message $\gamma$. Note that, $u\oplus_3\Pi(u)=0$.}\label{tab2}
\end{table}
\noindent Although \(\mathrm{C}^\circ(\mathcal{N}_3) = 0\), one may still ask with what probability Alice will succeed in sending her message \(\gamma \in \{0, 1\}\) to Bob using a single use of the channel \(\mathcal{N}_3\). Let Alice encodes her message \(\gamma = 0\) as the input \((0, 0)\) and \(\gamma = 1\) as the input \((1, 2)\) of the channel. At Bob's end, the outcome \((3, 0)\) is common to both inputs, which results in a success probability for decoding the message of \(\$(\mathcal{N}_3) = 7/8\). Due to the symmetries \((S_1)\) and \((S_2)\), it is evident that this is indeed the optimal success probability. With classical shared randomness this success probability remains same, i.e. \(\$_{SR}(\mathcal{N}_3)=7/8\). As established in our next result, with preshared entanglement the success \(\$_{E}(\mathcal{N}_3)\) can be higher.
\begin{proposition}\label{prop1}
A single use of the channel \(\mathcal{N}_3\) can successfully transmit 1 bit of information with a probability of \(0.9008\) when assisted by a two-qutrit maximally entangled state, and hence demonstrates that \(\$_E(\mathcal{N}_3) > \$_{SR}(\mathcal{N}_3)\).
\end{proposition}
\begin{proof}
(Outline) Given a generic NS correlation along with the channel \(\mathcal{N}_3\), various protocols could be designed for transmitting the one-bit message. Here, we adapt the protocol outlined in Table \ref{tab2} with minor modifications. When using the correlation \(P_3\) as assistance, this protocol allows Bob’s estimate \(\tilde{\gamma} \in \{0,1\}\) to exactly match Alice’s message \(\gamma\). However, for other NS correlations, \(\tilde{\gamma}\) may take values in \(\{0,1,2\}\). If \(\tilde{\gamma} = 2\), Bob makes a random guess of Alice’s message as either \(0\) or \(1\).

\noindent We employ the quantum correlation achieving the maximum violation of the two-input-three-output Collins-Gisin-Linden-Massar-Popescu (CGLMP) inequality, utilizing a maximally entangled two-qutrit state \cite{Collins2002}. Using this protocol, the resulting success probability becomes \(\$_E(\mathcal{N}_3) = \frac{1}{4}[1+\frac{1}{36} \text{C}^2(\frac{\pi}{4}) + \frac{1}{18} \text{C}^2(\frac{5\pi}{12}) + \frac{1}{6} \text{C}^2(\frac{\pi}{12})] \approx 0.9008 > 7/8\); here \(\text{C}(x)\equiv\text{Cosec}(x)\). This completes the proof with detailed calculations provided in Appendix \ref{appendix:b}.
\end{proof}
\noindent Notably, the success probability in Proposition \ref{prop1} might not be the optimal success achievable with quantum nonlocal correlations. Establishing optimality would require an exhaustive optimization over all possible protocols and all possible quantum-realizable correlations, which we leave as an open question for future research. Instead, we present here a generalization of the channel \(\mathcal{N}_3\). Consider the channel \(\mathcal{N}_m\) with $m\ge2$, which has inputs \((i_1, i_2) \in \{0, 1\} \times \{0, 1, \ldots, m-1\}\) and outputs \((o_1, o_2) \in \{1, 2, \ldots, m+1\} \times \{0, 1, \ldots, m-1\}\). The outcome \(o_2\) is conditioned on \(o_1\) according to the following rule:
\begin{align}
o_2 := 
\begin{cases} 
i_1, & \text{if } o_1 = 1; \\
i_2, & \text{if } o_1 = 2; \\
i_1 \oplus_m \Pi_m^{o_1}(i_2), & \text{if } o_1 \in \{3, \ldots, m+1\},
\end{cases}   
\end{align}
where \(\Pi_m^{o_1}(0) := 0\), and for \(i_2 \neq 0\), \(\Pi_m^{o_1}(i_2) := [(i_2 - 1) \oplus_{m-1} (o_1 - 3)] + 1\). For \(m = 3\), this recovers the channel \(\mathcal{N}_3\) from Eq. \((\ref{n3})\), and for \(m = 2\), it corresponds to the channel \(\mathcal{N}_{\boxtimes}\) defined in Ref. \cite{Prevedel2011}. The stochastic matrix of the channel \(\mathcal{N}_m\) is provided in Appendix \ref{appendix:a}, from which it is evident that the confusability graph \(\mathbb{G}(\mathcal{N}_m)\) corresponds to \(K_{2m}\), the complete graph with \(2m\) vertices. Consequently, we have \(\mathrm{C}_{SR}^\circ(\mathcal{N}_m) = 0\) for all values of \(m\). However, by utilizing the correlation \(P_m\), we can generalize Theorem \ref{theo1}.
\begin{theorem}\label{theo2}
For all $m\ge2$, with a single use of the channel \(\mathcal{N}_m\), one bit of information can be perfectly transmitted from the sender to the receiver when assisted by the NS correlation \(P_m\), resulting in \(\mathrm{C}^\circ_{NS}(\mathcal{N}_m) \geq 1\). 
\end{theorem}
\noindent The proof follows exactly as in Theorem \ref{theo1}, with Table \ref{tab1} being generalized to Table \ref{tab3}. Notably, a single copy of the correlation box \( P_m \) is sufficient to raise the zero-error capacity of the channel \( \mathcal{N}_{m^\prime} \) to one bit, provided \( m \ge m^\prime \). However, if \( m < m^\prime \), our protocol does not enable a nonzero capacity. On the other hand, while Theorem \ref{theo2} demonstrates a one-bit advantage in zero-error communication with NS correlation, this advantage can be made arbitrarily large by considering the channel \(\Lambda\equiv (\mathcal{N}_m)^{\otimes N}\). While \(\mathrm{C}^\circ(\Lambda) = 0\), when assisted by \(P_m^{\otimes N}\) the channel \(\Lambda\) can perfectly transmit \(N\)-bits of information.\\
\begin{table}[h!]
\begin{tabular}{|c|c|c|c|c|}
\hline
output $o_1$ & output $o_2$& ~~input $y~~$ &~~outcome $b$~~& ~~~~~~~~~~~guess $\tilde{\gamma}$~~~~~~~~~~~\\\hline
$1$ & $\gamma$ & $\times$ & $\times$ & $o_2$ \\\hline
$2$ & $a$ & $1$ & $\gamma\oplus_m a$ & $\Pi(o_{2})\oplus_m b$\\\hline
$l$ & $\gamma\oplus_m \Pi_m^{l}(a)$ & $0$ & $a$ & $o_2\oplus_m\hat{\Pi}_m(\Pi_m^{l}(b)) $\\\hline
\end{tabular}
\caption{Here, $l\in\{3,4,\cdots,m+1\}$; and $\hat{\Pi}_m:\{0,1,\cdots,m-1\}\to \{0,1,\cdots,m-1\}$ such that $u\oplus_m\hat{\Pi}_m(u)=0$, for all $u\in\{0,1,\cdots,m-1\}$.}\label{tab3}
\end{table}\\\\
We now demonstrate that the nonlocal advantage in zero-error communication can be made arbitrarily large without requiring multiple copies of the $2$-$2$-$m$ extremal NS correlation. Instead, this can be achieved using a single copy of $2$-$m$-$2$ extremal correlations. The complete characterization of such extremal correlation is established in \cite{Jones2005}. We consider a class of such extremal correlations \( \tilde{R}_m \equiv \{\tilde{r}_m(a, b | x, y)\} \) defined as:  
\begin{align}
\tilde{r}_m(a, b | x, y) := 
\begin{cases} 
\frac{1}{2}, ~~\text{if } a \oplus_2 b
= \sum_{i=1}^{m-1} \delta_{q,i} \cdot \delta_{\alpha,i}; \\     
0,~~\text{otherwise}.
\end{cases}
\end{align}
Consider now the following family of channels \( \mathcal{M}_m \equiv \{p(o_1, o_2 | i_1, i_2)\} \) with inputs \( (i_1, i_2) \in \{0, 1, \ldots, m-1\} \times \{0, 1\} \) and outputs \( (o_1, o_2) \in \{1, 2, \ldots, m(m-1)+1\} \times \{0, 1, \ldots, m-1\} \), where \( m \geq 2 \). The output \( o_2 \) is conditioned on \( o_1 \) as follows:
\begin{align}
o_2 :=
\begin{cases} 
i_1, ~~~~\text{if } o_1 = 1; \\
i_1 \oplus_m \Pi_m^{o_1-2}(i_2), ~~~~ \text{if } o_1 \in \{2, \cdot\cdot, m\}; \\
i_1 \oplus_m \Pi_m^{o_1-(m+1)}(i_2), \\
\text{if } o_1 \in \{m+1, \cdot\cdot, 2m-1\}; \\
\vdots \\
i_1 \oplus_m \Pi_m^{o_1-(m-1)^2}(i_2), \\
\text{if } o_1 \in \{(m-1)^2+1, \cdot\cdot, m(m-1)+1\};
\end{cases}
\end{align}
where \( \Pi_m^{o_1-k}(0) := 0 \), and for \( i_2 \neq 0 \), \( \Pi_m^{o_1-k}(i_2) := [(i_2 - 1) \oplus_{m-1} (o_1 - k)] + 1 \). Using this class of channels, we establish our next result.
\newpage
\begin{figure}[h!]
\centering
\begin{minipage}{.4\textwidth}
\centering
\begin{tikzpicture}
\node[shape=circle,draw=black,fill=blue!5] (00) at (0,5.2) {$(0,0)$};
\node[shape=circle,draw=black,fill=blue!5] (01) at (3,5.2) {$(0,1)$};
\node[shape=circle,draw=black,fill=blue!5] (02) at (4.5,2.6) {$(0,2)$};
\node[shape=circle,draw=black,fill=blue!5] (10) at (3,0) {$(1,0)$};
\node[shape=circle,draw=black,fill=blue!5] (11) at (0,0) {$(1,1)$};
\node[shape=circle,draw=black,fill=blue!5] (12) at (-1.5,2.6) {$(1,2)$};
\node[] at (1.5,5.4) {$\textcolor{purple}{(1,0)}$};
\node[] at (3.95,3.9) {\rotatebox{300}{$\textcolor{purple}{(1,0)}$}};
\node[] at (1.15,4.7) {$\rotatebox{330}{$\textcolor{purple}{(1,0)}$}$};
\node[] at (1.5,-0.2) {$\textcolor{blue}{(1,1)}$};
\node[] at (-0.8,1.0) {\rotatebox{300}{$\textcolor{blue}{(1,1)}$}};
\node[] at (0.4,1.7) {\rotatebox{330}{$\textcolor{blue}{(1,1)}$}};
\node[] at (1.1,3.6) {\rotatebox{-60}{${(2,0)}$}};
\node[] at (1.9,3.6) {\rotatebox{60}{${(2,1)}$}};
\node[] at (2.3,2.75) {\rotatebox{0}{${(2,2)}$}};
\node[] at (-1,3.9) {\rotatebox{-300}{${(3,0)}$}};
\node[] at (3.15,1.2) {\rotatebox{-90}{${(3,1)}$}};
\node[] at (1.8,1.2) {\rotatebox{-330}{${(3,2)}$}};
\node[] at (0.2,3.05) {\rotatebox{-90}{${(4,0)}$}};
\node[] at (3.9,1.2) {\rotatebox{-300}{${(4,1)}$}};
\node[] at (-0.5,3.4) {\rotatebox{30}{${(4,2)}$}};
\path [purple] (00) edge node[left] {} (01);
\path [purple] (00) edge node[left] {} (02);
\path [] (00) edge node[left] {} (10);
\path [] (00) edge node[left] {} (11);
\path [] (00) edge node[left] {} (12);
\path [purple] (01) edge node[left] {} (02);
\path [] (01) edge node[left] {} (10);
\path [] (01) edge node[left] {} (11);
\path [] (01) edge node[left] {} (12);
\path [] (02) edge node[left] {} (10);
\path [] (02) edge node[left] {} (11);
\path [] (02) edge node[left] {} (12);
\path [blue] (10) edge node[left] {} (11);
\path [blue] (10) edge node[left] {} (12);
\path [blue] (11) edge node[left] {} (12);
\end{tikzpicture}
\end{minipage}
\hspace{1.5cm}
\begin{minipage}{.4\textwidth}
\centering
\begin{tikzpicture}
\node[shape=circle,draw=black,fill=purple!5] (00) at (0,5.2) {$(0,0)$};
\node[shape=circle,draw=black,fill=purple!5] (01) at (3,5.2) {$(0,1)$};
\node[shape=circle,draw=black,fill=purple!5] (02) at (4.5,2.6) {$(2,1)$};
\node[shape=circle,draw=black,fill=purple!5] (10) at (3,0) {$(1,0)$};
\node[shape=circle,draw=black,fill=purple!5] (11) at (0,0) {$(1,1)$};
\node[shape=circle,draw=black,fill=purple!5] (12) at (-1.5,2.6) {$(2,0)$};
\node[] at (1.5,5.4) {$(1,0)$};
\node[] at (3.95,3.9) {\rotatebox{300}{$(7,2)$}};
\node[] at (1.15,4.7) {\rotatebox{333}{$(2,0)^\star$}};
\node[] at (1.5,-0.2) {$(1,1)$};
\node[] at (-0.8,1.0) {\rotatebox{300}{$(2,2)$}};
\node[] at (0.4,1.7) {\rotatebox{330}{$(4,2)^\star$}};
\node[] at (1.1,3.6) {\rotatebox{-60}{$(5,0)$}};
\node[] at (1.9,3.6) {\rotatebox{60}{$(4,1)$}};
\node[] at (2.3,2.75) {\rotatebox{0}{$(1,2)$}};
\node[] at (-1,3.9) {\rotatebox{-300}{$(6,0)$}};
\node[] at (3.15,1.2) {\rotatebox{-90}{${(2,1)}^\star$}};
\node[] at (1.8,1.2) {\rotatebox{-330}{$(5,1)^\star$}};
\node[] at (0.2,3.05) {\rotatebox{-90}{$(3,0)^\star$}};
\node[] at (3.9,1.2) {\rotatebox{-300}{$(3,1)$}};
\node[] at (-0.5,3.4) {\rotatebox{30}{$(3,2)^\star$}};
\path [] (00) edge node[left] {} (01);
\path [] (00) edge node[left] {} (02);
\path [] (00) edge node[left] {} (10);
\path [] (00) edge node[left] {} (11);
\path [] (00) edge node[left] {} (12);
\path [] (01) edge node[left] {} (02);
\path [] (01) edge node[left] {} (10);
\path [] (01) edge node[left] {} (11);
\path [] (01) edge node[left] {} (12);
\path [] (02) edge node[left] {} (10);
\path [] (02) edge node[left] {} (11);
\path [] (02) edge node[left] {} (12);
\path [] (10) edge node[left] {} (11);
\path [] (10) edge node[left] {} (12);
\path [] (11) edge node[left] {} (12);
\end{tikzpicture}
\end{minipage}
\caption{Left: Confusability graph \(\mathbb{G}(\mathcal{N}_3)\). Right: Confusability graph \(\mathbb{G}(\mathcal{M}_3)\). Nodes represent the inputs, while edges denote outputs that induce confusion between the connected nodes. Despite both graphs being \(K_6\), the channels \(\mathcal{N}_3\) and \(\mathcal{M}_3\) exhibit a key difference. In \(\mathbb{G}(\mathcal{N}_3)\), certain outputs cause confusion among more than two inputs. For example, the output tuple \((1, l)\) results in confusability among the inputs \(\{(l,0), (l,1), (l,2)\}\), where \textcolor{blue}{$l=1$} and \textcolor{purple}{$l=0$}, respectively. This behavior is absent in \(\mathbb{G}(\mathcal{M}_3)\). Instead, in this case, multiple output tuples may lead to confusability between the same pair of inputs. For instance, the inputs \((0,0)\) and \((2,1)\) are confusable for both the outputs \((2,0)\) and \((4,0)\). Such cases are marked with an additional star, and the corresponding outputs are detailed in Table \ref{tab6}.}\label{fig4}
\vspace{-.5cm}
\end{figure}
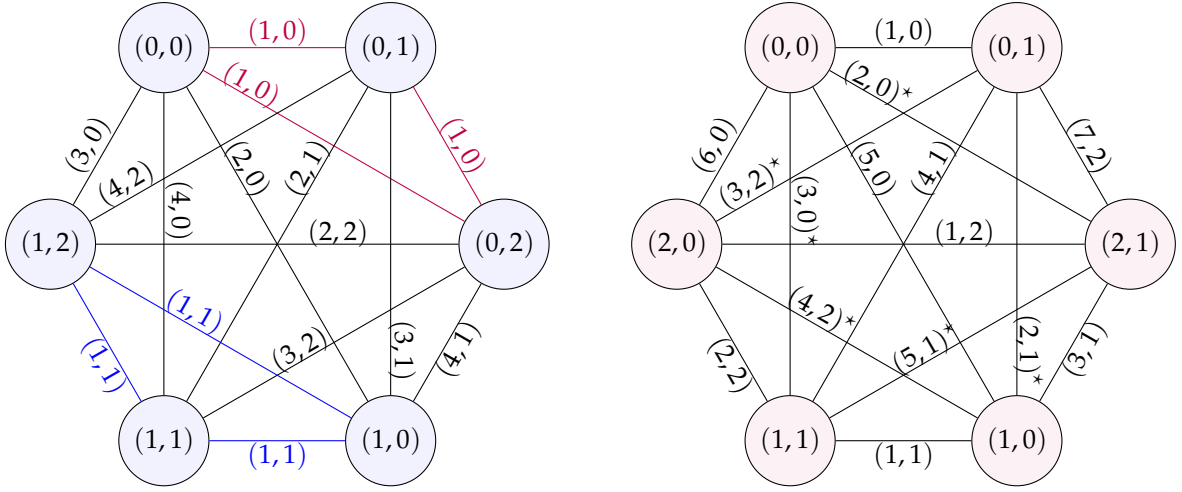
\begin{theorem}\label{theo3}
For all \( m \geq 2 \), we have \( \mathrm{C}^\circ_{SR}(\mathcal{M}_m) = 0 \). However, in assistance with the correlation \( \tilde{R}_m \), \(\log m\)-bits of information can be perfectly transmitted through the channel $\mathcal{M}_m$, implying that \( \mathrm{C}^\circ_{NS}(\mathcal{M}_m) \geq \log m \). \end{theorem}
\begin{proof}
The confusability graph \( \mathbb{G}(\mathcal{M}_m) \) of the channel \( \mathcal{M}_m \) is a complete graph \( K_{2m} \) with \( 2m \) vertices (see Appendix \ref{appendix:c}). Consequently, \( \mathrm{C}^\circ_{SR}(\mathcal{M}_m) = 0 \). However, with the assistance of the NS correlation \( \tilde{R}_m \), Alice can transmit \( \log m \)-bits of information perfectly to Bob by employing a protocol analogous to Theorem \ref{theo2}. Given a message \(\gamma \in \{0, 1, \cdots, m-1\}\), Alice chooses her input \(x\) for the NS box \(\tilde{R}_m\) as \(\gamma\) and obtains the outcome \(a\). She then feeds the pair \((\gamma, a)\) as the input \((i_1, i_2)\) to the channel \(\mathcal{M}_m\). The receiver, Bob, selects his input \(y\) for his part of the NS box \(\tilde{R}_m\) based on the outcome \((o_1, o_2)\) received from \(\mathcal{M}_m\) and obtains an outcome \(b\). He then guesses \(\tilde{\gamma}\) depending on \((o_1, o_2)\) and \(b\). The protocol is illustrated in Fig. \ref{fig3}, and the calculation in Table \ref{tab4} demonstrates that Bob always guesses Alice’s message correctly. 
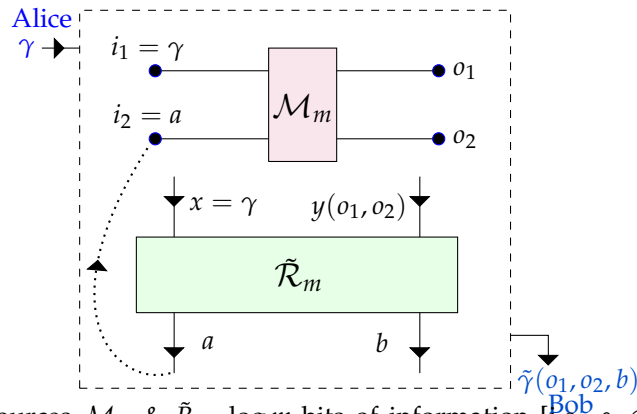
\begin{figure}[h!]
\centering
\begin{tikzpicture}
\draw[fill=purple!10](0,0)--(0,1.5)--(.90,1.5)--(.90,0)--(0,0);
\draw[thin,black](-1.5,1.2)--(0,1.2);
\draw[thin,black](-1.5,.3)--(0,0.3);
\draw[thin,black](.90,1.2)--(2.25,1.2);
\draw[thin,black](.90,.3)--(2.25,0.3);
\draw[blue, fill=black] (-1.5,1.2) circle (.08);
\draw[blue, fill=black] (-1.5,.3) circle (.08);
\draw[blue, fill=black] (2.25,1.2) circle (.08);
\draw[blue, fill=black] (2.25,0.3) circle (.08);
\node[] at (.45,.7){\large{$\mathcal{M}_m$}};
\node[] at (-1.6,1.5){$i_1=\gamma$};
\node[] at (-1.6,.6){$i_2=a$};
\node[] at (2.6,1.2){$o_1$};
\node[] at (2.6,.3){$o_2$};

\draw[fill=green!10](-1.75,-2)--(-1.75,-1)--(2.5,-1)--(2.5,-2)--(-1.75,-2);
\draw[thin,black](-1.25,-.2)--(-1.25,-1);
\draw[thin,black](-1.25,-2)--(-1.25,-2.8);
\draw[thin,black](2,-.2)--(2,-1);
\draw[thin,black](2,-2)--(2,-2.8);
\node[] at (-1.25,-.6)[rotate=180,inner sep=0pt] {\tikz\draw[-triangle 90](0,0) ;};
\node[] at (-1.25,-2.6)[rotate=180,inner sep=0pt] {\tikz\draw[-triangle 90](0,0) ;};
\node[] at (2,-.6)[rotate=180,inner sep=0pt] {\tikz\draw[-triangle 90](0,0) ;};
\node[] at (2,-2.6)[rotate=180,inner sep=0pt] {\tikz\draw[-triangle 90](0,0) ;};
\draw [thick,dotted] (-1.3,-2.8) .. controls (-2,-3) and (-3,-2) .. (-1.6,.18);

\node[] at (-2.26,-1.3)[rotate=-5,inner sep=0pt] {\tikz\draw[-triangle 90](0,0) ;};
\node[] at (.4,-1.5){\large{$\mathcal{\tilde{R}}_m$}};
\node[] at (-.6,-.6){$x=\gamma$};
\node[] at (1.2,-.6){$y(o_1,o_2)$};
\node[] at (-.8,-2.4){$a$};
\node[] at (1.5,-2.4){$b$};

\draw[thin,black, dashed](-2.5,-3)--(-2.5,2)--(3.2,2)--(3.2,-3)--(-2.5,-3);
\draw[thin,black](-3,1.5)--(-2.5,1.5);
\draw[thin,black](3.2,-2.3)--(3.7,-2.3)--(3.7,-2.6);
\node[] at (-2.7,1.5)[rotate=270,inner sep=0pt] {\tikz\draw[-triangle 90](0,0) ;};
\node[] at (3.7,-2.7)[rotate=180,inner sep=0pt] {\tikz\draw[-triangle 90](0,0) ;};
\node[] at (-3.2,1.5){\textcolor[rgb]{0,0,1}{$\gamma$}};
\node[] at (4.1,-2.9){\textcolor[rgb]{0,.3,.8}{$\tilde{\gamma}(o_1,o_2,b)$}};
\node[] at (-3.,1.9){\textcolor[rgb]{0,0,1}{Alice}};
\node[] at (4,-3.2){\textcolor[rgb]{0,.3,.8}{Bob}};
\end{tikzpicture}
\vspace{-.73cm}
\caption{Given the resources $\mathcal{M}_m~\&~\tilde{R}_m$, $\log m$ bits of information [i.e. $\gamma\in\{0,1\cdots,m-1\}$] can be transmitted perfectly from Alice to Bob.}\label{fig3}
\end{figure}
\begin{table}[h!]
\begin{tabular}{|c|c|c|c|c|}
\hline
output $o_1$ & output $o_2$& ~~input $y$~~ &outcome $b$& ~~~~guess $\tilde{\gamma}$~~~~\\\hline
$1$ & $\gamma$ & $\times$ & $\times$ & $o_2$ \\\hline
$2$ & $\gamma \oplus_m \Pi^{0}_m(a)$ & $0$ & $ a$ & $o_2\oplus_m \hat{\Pi}_{m}(\Pi^0_m(b))$\\\hline
$3$ & $\gamma\oplus_m\Pi^{1}_m(a)$ & $0$ & $ a$ & $o_2\oplus_m \hat{\Pi}_{m}(\Pi^{1}_m(b))$\\\hline
\vdots &\vdots &\vdots &\vdots &\vdots \\\hline
$m$ & $\gamma\oplus_m \Pi_m^{m-2}(a)$ & $0$ & $a$ & $o_2\oplus_m\hat{\Pi}_m(\Pi_m^{m-2}(b)) $\\\hline
$m+1$ & $\gamma \oplus_m \Pi^{0}_m(a\oplus_2\delta_{\gamma,1})$ & $1$ & $ a\oplus_2\delta_{\gamma,1} $ & $o_2\oplus_m \hat{\Pi}_{m}(\Pi^0_m(b))$\\\hline
\vdots &\vdots &\vdots &\vdots &\vdots \\\hline
$2m-1$ & $\gamma \oplus_m \Pi^{m-2}_m(a\oplus_2\delta_{\gamma,1})$ & $1$ & $ a\oplus_2\delta_{\gamma,1} $ & $o_2\oplus_m \hat{\Pi}_{m}(\Pi^{m-2}_{m}(b))$\\\hline
\vdots &\vdots &\vdots &\vdots &\vdots \\\hline
\vdots &\vdots &\vdots &\vdots &\vdots \\\hline
$(m-1)^2+1$ & $\gamma \oplus_m \Pi^{0}_m(a\oplus_2\delta_{\gamma,m-1})$ & $m-1$ & $ a\oplus_2\delta_{\gamma,m-1} $ & $o_2\oplus_m \hat{\Pi}_{m}(\Pi^0_m(b))$\\\hline
\vdots &\vdots &\vdots &\vdots &\vdots \\\hline
$m(m-1)+1$ & $\gamma \oplus_m \Pi^{m-2}_m(a\oplus_2\delta_{\gamma,m-1})$ & $m-1$ & $ a\oplus_2\delta_{\gamma,m-1} $ & ~~$o_2\oplus_m \hat{\Pi}_{m}(\Pi^{m-2}_{m}(b))$~~\\\hline
\end{tabular}
\caption{Here, $l\in\{3,4,\cdots,m+1\}$; and $\hat{\Pi}_m:\{0,1,\cdots,m-1\}\to \{0,1,\cdots,m-1\}$ such that $u\oplus_m\hat{\Pi}_m(u)=0$, for all $u\in\{0,1,\cdots,m-1\}$.}\label{tab4}
\end{table}
\end{proof}
\noindent Since all extremal NS correlations in the $2$-$m$-$2$ scenario contain the $2$-$2$-$2$ Popescu-Rohrlich correlation \cite{Jones2005}, they are all capable of enabling the perfect transmission of 1-bit of information through the channel \( \mathcal{M}_2 \), which is identical to the channel \( \mathcal{N}_\boxtimes \) previously studied in \cite{Prevedel2011}. A result similar to Proposition \ref{prop1} also follows in this case (the proof provided in Appendix \ref{appendix:d}).
\begin{proposition}\label{prop2}
A single use of the channel \(\mathcal{M}_3\) can successfully transmit $\log3$ bits of information with a probability of \(6/7\) when assisted by a two-qubit maximally entangled state, and hence demonstrates that \(\$_E(\mathcal{M}_3) > \$_{SR}(\mathcal{M}_3)=17/21\).
\end{proposition}
\noindent Notably, using the channel \( \mathcal{M}_m \), Alice can transmit at least \(\log m\)-bits of information perfectly to Bob if they share the extremal correlation \( \tilde{R}_{m^\prime} \) with \( m^\prime \geq m \). The next theorem shows that this is the optimal Alice can transmit to Bob if they share any $2$-$m$-$2$ NS correlation.
\begin{theorem}\label{theo4}
The one-shot zero-error capacity of any channel \( \mathcal{M} \) with \(\mathrm{C}^\circ(\mathcal{M}) = 0\) assisted with any $2$-$m$-$2$ no-signaling correlation is at most \(\log m\) bits.
\end{theorem}
\begin{proof}
Let's assume there exists an $2$-$m$-$2$ correlation for  \(\mathrm{C}^\circ_{2\text{-}m\text{-}2}(\mathcal{M}) > \log m\). Consider channel  \(\mathcal{M}\) assisted with a \(\log m\) identity channel. By the superadditivity of zero-error capacities, the zero-error capacity of the combined channel is \(\log m\). However, since \(\log m\) bits are sufficient to simulate any $2$-$m$-$2$ correlation, this setup is equivalent to  \( \mathcal{M} \) assisted by the $2$-$m$-$2$  correlation, which by assumption has a zero-error capacity strictly greater than \( \log m \), contradicting our earlier bound.  Thus, no such $2$-$m$-$2$ correlation exists such that \(\mathrm{C}^\circ_{2\text{-}m\text{-}2}(\mathcal{M}) > \log m\).
\end{proof}
So far we observed that channels with  \(\mathrm{C}^\circ(\mathcal{M}) =0\) can sometimes be made useful with the assistance of a no-signaling correlation, say $R_{NS}$, i.e, \(\mathrm{C}^\circ_{R_{NS}}(\mathcal{M}) = \log(m)>0\). We now show that this enhanced capacity $\log(m)$ is also related to the classical communication cost of simulating the non-local correlation, $R_{NS}$.
\begin{theorem}\label{theo5}
If for a channel \(\mathcal{M}\) we have \(\mathrm{C}^\circ(\mathcal{M}) =0\) and  \(\mathrm{C}^\circ_{R_{NS}}(\mathcal{M}) = \log(m)>0\) for some no signaling correlation $R_{NS}$, then \(\log(m)\) bits of classical communication is necessary to simulate $R_{NS}$.
\end{theorem}
\begin{proof}
Let's assume that the correlation $R_{NS}$ can be simulated by $\log (m-1)$ bits. This implies the channel \(\mathcal{M}\) assisted with an $\log (m-1)$ identity channel can achieve zero error capacity of $\log(m)$ simply by using $\log (m-1)$ bits to simulate $R_{NS}$. But this violates the strong product rule. Thus  \(\log(m)\) bits of classical communication are necessary to simulate $R_{NS}$.
\end{proof}
\section{Discussion}
\noindent The zero-error capacity of a channel exhibits intriguing behaviors distinct from its standard capacity \cite{Shannon1948}. For instance, as predicted by Shannon \cite{Shannon1956} and  then established by Lovász \cite{Lovasz1979}, the zero-error capacity of noisy classical channels can exhibit superadditivity phenomenon, where the combined usage of channels surpasses the sum of their individual utilities. More recently several other intriguing features of zero-error capacities have been reported by considering noisy quantum channels \cite{Duan2009,Cubitt2011(1),Chen2010}. Additionally, unlike the standard capacity, the zero-error capacity of noisy classical channels can be enhanced through the sharing of nonlocal correlations between the sender and receiver \cite{Cubitt2011}. In this work, we demonstrate that all extremal no-signaling (NS) correlations in the $2$-$2$-$m$ and $2$-$m$-$2$ Bell scenarios exhibit exotic utility in zero-error communication theory. Specifically, these correlations enable the perfect transmission of classical information through noisy classical channels that otherwise have zero zero-error capacity. While nonlocal correlations obtained in quantum theory do not achieve such extreme functionality, we show that they can significantly enhance the success probability of transmitting information through these channels.

\noindent Our findings also open up several avenues for further research. It is known that none of the extremal nonlocal correlations in generic Bell scenarios (with arbitrary but finite input-output cardinalities) are quantum-realizable \cite{Ramanathan2016}. An interesting direction would be to establish whether every such extremal correlation enables communication through noisy classical channels with zero zero-error capacity. On the other hand, characterizing the specific quantum correlations that improve zero-error communication performance, as in Propositions \ref{prop1} \& \ref{prop2}, remains an open question. Finally, experimentally demonstrating the reported quantum advantage in zero-error communication setups would provide a valuable step toward practical applications.\\

\noindent {\bf Acknowledgments:} KA acknowledges support from the CSIR project $09/0575(19300)/2024$-EMR-I. SGN acknowledges support from the CSIR project $09/0575(15951)/2022$-EMR-I. MB acknowledges funding from the National Quantum Mission.

\newpage

\appendix
\section{Stochastic matrix of the channel $\mathcal{N}_m$} \label{appendix:a}
\noindent For generic $m~(\ge 2)$ the channel $\mathcal{N}_m$ looks as:
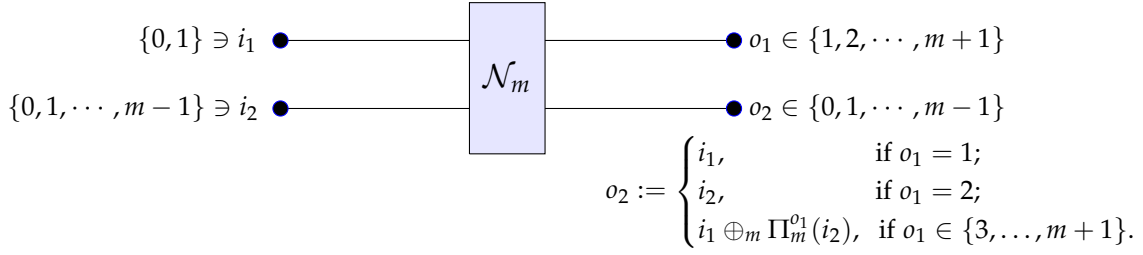
\begin{figure}[h!]
\centering
\begin{tikzpicture}
\draw[fill=blue!10](-2,0)--(-2,2)--(-1,2)--(-1,0)--(-2,0);
\draw[thin,black](-4.5,1.5)--(-2,1.5);
\draw[thin,black](-4.5,.6)--(-2,.6);
\draw[thin,black](-1,1.5)--(1.5,1.5);
\draw[thin,black](-1,.6)--(1.5,.6);
\draw[blue, fill=black] (-4.5,1.5) circle (.1);
\draw[blue, fill=black] (-4.5,.6) circle (.1);
\draw[blue, fill=black] (1.5,1.5) circle (.1);
\draw[blue, fill=black] (1.5,.6) circle (.1);
\node[] at (-1.5,1){\large{$\mathcal{N}_m$}};
\node[] at (-5.6,1.5){$\{0,1\}\ni i_1$};
\node[] at (-6.45,.6){$\{0,1,\cdots,m-1\}\ni i_2$};
\node[] at (3.4,1.5){$o_1\in\{1,2,\cdots,m+1\}$};
\node[] at (3.4,.6){$o_2\in\{0,1,\cdots,m-1\}$};
\node[] at (3.5,-.5)
{$o_2:=\begin{cases} 
i_1,~~~~~~~~~~~~~~~~~~~\text{if } o_1 = 1;\\
i_2,~~~~~~~~~~~~~~~~~~~\text{if } o_1 = 2;\\
i_1 \oplus_m \Pi_m^{o_1}(i_2),~~ \text{if } o_1 \in \{3, \ldots, m+1\}.
\end{cases}$};
\end{tikzpicture}
\end{figure}\\
The stochastic matrix \(\mathbb{N}_m\equiv\{\text{Pr}(o_1, o_2 | i_1, i_2)\}\) is given by,
\begin{table}[h]
\begin{tabular}
{|c||c|c|c|c|c|c||c|c|c|c|c|c|c|}
\hline
$\text{Output} \backslash \text{Input}$ & ~$(0,0)$~ & ~$(0,1)$~ &~$(0,2)$~&~~~$\cdot$~~~~~&~~~~$\cdot$~~~~~&~$(0,m-1)$~&~$(1,0)$~ & ~$(1,1)$&~$(1,2)$~&~~~$\cdot$~~~~~&~~~~$\cdot$~~~~~&~$(1,m-1)$~ \\ \hline\hline
$\textcolor{purple}{(1,0)}$ & ~$\textcolor{purple}{\Omega_m}$~ & ~$\textcolor{purple}{\Omega_m}$~ &~$\textcolor{purple}{\Omega_m}$~&~~~~$\cdot$~~~~~&~$\cdot$~&~$\textcolor{purple}{\Omega_m}$~ &~$0$~ & ~$0$&~~~~$0$~~~~~&~~~~$\cdot$~~~~~&~~~~$\cdot$~~~~~&~$0$~ \\ \hline
$\textcolor{blue}{(1,1)}$ & ~$0$~ & ~$0$~ &~~~~$0$~~~~~&~$\cdot$~&~~~~$\cdot$~~~~~&~$0$~&~$\textcolor{blue}{\Omega_m}$~ & ~$\textcolor{blue}{\Omega_m}$&~~~~$\textcolor{blue}{\Omega_m}$~~~~~&~~~~$\cdot$~~~~~&~~~~$\cdot$~~~~~&~$\textcolor{blue}{\Omega_m}$~ \\ \hline
$(2,0)$ & $\Omega_m$ & $0$ & ~~~~$0$~~~~~&$\cdot$& $\cdot$& $0$& $\Omega_m$ & $0$&$0$&$\cdot$ & $\cdot$&$0$ \\ \hline
$(2,1)$ & $0$ & $\Omega_m$ & ~~~~$0$~~~~~&$\cdot$& $\cdot$&$0$ &$0$ & $\Omega_m$ & $0$&$\cdot$ &$\cdot$ &$0$ \\ \hline
$(2,2)$ & $0$ & $0$ $\cdot$ &$\Omega_m$ & ~~~~$\cdot$~~~~~&~$\cdot$&$0$ &$0$ & $0$& $\Omega_m$ & $\cdot$ & $\cdot$& $0$\\ \hline
$\cdot$ & $\cdot$ & $\cdot$ & ~~~~$\cdot$~~~~~&$\Omega_m$& $\cdot$& $\cdot$& $\cdot$& $\cdot$&$\cdot$& $\Omega_m$&$\cdot$ &$\cdot$ \\ \hline
$\cdot$ & $\cdot$ & $\cdot$ & ~~~~$\cdot$~~~~~&$\cdot$& $\Omega_m$& $\cdot$& $\cdot$& $\cdot$& $\cdot$&$\cdot$&$\Omega_m$ &$\cdot$ \\ \hline
$(2,m-1)$ & $0$ & $0$ &$0$&$\cdot$ &$\cdot$ & $\Omega_m$ & $0$&$0$ &$0$ & $\cdot$&$\cdot$& $\Omega_m$\\ \hline
$\cdot$ & $\cdot$ & $\cdot$ & $\cdot$& $\cdot$&$\cdot$& $\cdot$& $\cdot$& $\cdot$& $\cdot$&$\cdot$ &$\cdot$&$\cdot$ \\ \hline
$\cdot$ & $\cdot$ & $\cdot$ & $\cdot$&$\cdot$& $\cdot$& $\cdot$& $\cdot$& $\cdot$& $\cdot$&$\cdot$ &$\cdot$&$\cdot$ \\ \hline
$\cdot$ & $\cdot$ & $\cdot$ & $\cdot$&$\cdot$& $\cdot$& $\cdot$& $\cdot$& $\cdot$& $\cdot$&$\cdot$ &$\cdot$&$\cdot$ \\ \hline
$(m+1,0)$ & $\Omega_m$  &$0$  &$0$ &$\cdot$ &$\cdot$ & $0$&$0$&$\Omega_m$ & $0$&$\cdot$& $\cdot$&$0$ \\ \hline
$(m+1,1)$ & $0$ &$0$  & $\Omega_m$&$\cdot$& $\cdot$& $0$& $\Omega_m$&$0$ &$0$ &$\cdot$ &$\cdot$&$0$ \\ \hline
$(m+1,2)$ & $0$ &$0$  &$0$& $\Omega_m$&$\cdot$&$0$ &$0$ &$0$& $\Omega_m$& $\cdot$& $\cdot$&$0$  \\ \hline
$\cdot$ & $\cdot$ & $\cdot$ & $\cdot$&$\cdot$& $\Omega_m$& $\cdot$& $\cdot$& $\cdot$& $\cdot$&$\Omega_m$&$\cdot$ &$\cdot$ \\ \hline
$\cdot$ & $\cdot$ & $\cdot$ & $\cdot$& $\cdot$&$\cdot$& $\Omega_m$& $\cdot$& $\cdot$& $\cdot$&$\cdot$ &$\Omega_m$&$\cdot$ \\ \hline
$(m+1,m-1)$ & $0$ & $\Omega_m$ &$0$ &$\cdot$ &$\cdot$ & $0$&$0$ &$0$ &$0$ &$\cdot$&$\cdot$& $\Omega_m$\\ \hline
\end{tabular}
\caption{The stochastic matrix \(\mathbb{N}_m \equiv \{\text{Pr}(o_1, o_2 | i_1, i_2)\}\) o the channel \(\mathcal{N}_m\), with \(\Omega_m := 1/(m+1)\). All inputs in the set \(\mathcal{I}_l := \{(l, k)\}_{k=0}^{m-1}\) share \((1, l)\) as a common possible outcome for \textcolor{purple}{$l=0$} and \textcolor{blue}{$l=1$}, making them mutually confusable. Conversely, a pair of inputs, one from \textcolor{purple}{$\mathcal{I}_0$} and the other from \textcolor{blue}{$\mathcal{I}_1$}, share exactly one common outcome. Importantly, the common outcome is distinct for each such pair. Thus, for the channel \(\mathcal{N}_m\), the total of \(2 + m^2\) outcomes ensures that the confusability graph \(\mathbb{G}(\mathcal{N}_m)\) forms a complete graph \(K_{2m}\).}\label{tab5}
\end{table}\\

\section{Advantage of Entanglement Assistance in $\mathcal{N}_3$ Channel [Proposition 1]}\label{appendix:b}
\noindent The NS correlation $P_3$ utilized in Theorem 1 reads as
\begin{align}
P_3\equiv\begin{array}{c||c|c|c|c|c|c|c|c|c|}
\hline
x,y\backslash a,b & 0,0 & 0,1& 0,2& 1,0& 1,1& 1,2& 2,0& 2,1& 2,2 \\\hline\hline
0,0& 1/3& 0 & 0 & 0 & 1/3 & 0 & 0 & 0 &1/3\\\hline
0,1& 1/3& 0 & 0 & 0 & 1/3 & 0 & 0 & 0 &1/3\\\hline
1,0& 1/3& 0 & 0 & 0 & 1/3 & 0 & 0 & 0 &1/3\\\hline
1,1& 0& 1/3 & 0 & 0 & 0 & 1/3 & 1/3 & 0 &0\\\hline
\end{array}~.  
\end{align}
\noindent Nonlocality of this correlation can be established by the violation of two-input-three-output bipartite Bell inequality ($\mathbb{I}_3$) proposed by Collins-Gisin-Linden-Massar-Popescu (CGLMP) \cite{Collins2002}. Notably, the correlation $P_3$ saturates the algebraic maximum of the inequality $\mathbb{I}_3$. The authors in \cite{Collins2002} also analyzed the optimal violation of this inequality with the maximally entangled state $\ket{\phi^+}:=\frac{1}{\sqrt{3}}(\ket{00}+\ket{11}+\ket{22})\in\mathbb{C}^3\otimes\mathbb{C}^3$. The resulting correlation reads as: 
\begin{align}\hspace{-1cm}
P^{\phi^+}_3\equiv\begin{array}{c||c|c|c|c|c|c|c|c|c|}
\hline
x,y\backslash a,b & 0,0 & 0,1& 0,2& 1,0& 1,1& 1,2& 2,0& 2,1& 2,2 \\\hline\hline
0,0& \eta\text{C}^2\left(\frac{\pi}{12}\right)& \eta\text{C}^2\left(\frac{5\pi}{12}\right) & \eta\text{C}^2\left(\frac{\pi}{4}\right) & \eta\text{C}^2\left(\frac{\pi}{4}\right) & \eta\text{C}^2\left(\frac{\pi}{12}\right) & \eta\text{C}^2\left(\frac{5\pi}{12}\right) & \eta\text{C}^2\left(\frac{5\pi}{12}\right) & \eta\text{C}^2\left(\frac{\pi}{4}\right) &\eta\text{C}^2\left(\frac{\pi}{12}\right)\\\hline
0,1 & \eta\text{C}^2\left(\frac{\pi}{12}\right)  & \eta\text{C}^2\left(\frac{\pi}{4}\right) & \eta\text{C}^2\left(\frac{5\pi}{12}\right) & \eta\text{C}^2\left(\frac{5\pi}{12}\right) & \eta\text{C}^2\left(\frac{\pi}{12}\right)  & \eta\text{C}^2\left(\frac{\pi}{4}\right)   & \eta\text{C}^2\left(\frac{\pi}{4}\right)   & \eta\text{C}^2\left(\frac{5\pi}{12}\right) &\eta\text{C}^2\left(\frac{\pi}{12}\right)\\\hline
1,0 & \eta\text{C}^2\left(\frac{\pi}{12}\right)  & \eta\text{C}^2\left(\frac{\pi}{4}\right)   & \eta\text{C}^2\left(\frac{5\pi}{12}\right) & \eta\text{C}^2\left(\frac{5\pi}{12}\right) & \eta\text{C}^2\left(\frac{\pi}{12}\right)  & \eta\text{C}^2\left(\frac{\pi}{4}\right)   & \eta\text{C}^2\left(\frac{\pi}{4}\right)   & \eta\text{C}^2\left(\frac{5\pi}{12}\right) &\eta\text{C}^2\left(\frac{\pi}{12}\right)\\\hline
1,1 & \eta\text{C}^2\left(\frac{\pi}{4}\right)   & \eta\text{C}^2\left(\frac{\pi}{12}\right)  & \eta\text{C}^2\left(\frac{5\pi}{12}\right) & \eta\text{C}^2\left(\frac{5\pi}{12}\right) & \eta\text{C}^2\left(\frac{\pi}{4}\right)   & \eta\text{C}^2\left(\frac{\pi}{12}\right)  & \eta\text{C}^2\left(\frac{\pi}{12}\right)  & \eta\text{C}^2\left(\frac{5\pi}{12}\right) & \eta\text{C}^2\left(\frac{\pi}{4}\right)\\\hline
\end{array}   \label{mec}
\end{align}
where $\eta=1/54$ and $\text{C}(x)\equiv\text{Cosec}(x)$.\\\\
While sending the information $\gamma\in\{0,1\}$ through the channel $\mathcal{N}_3$ in assistance with a generic NS correlation by employing the strategy of Table 1 in the main manuscript, Bob's guess $\tilde{\gamma}$ belongs to the set $\{0,1,2\}$. For the case $\tilde{\gamma}=2$, Bob can randomly modify his guess as $0$ and $1$. Alternatively, since $\gamma$ is uniformly random over $\{0,1\}$, Bob can choose to answer $\tilde{\gamma}$ as $0$ whenever $\tilde{\gamma}=2$. Denoting $x_{\oplus_{2}}:= x ~\text{mod}~ 2$, the success probability of the protocol reads as  
\begin{subequations}
\begin{align}
\$_{NS}(\mathcal{N}_3)&=p(\tilde{\gamma} = \gamma ,\gamma)
=\sum_{\gamma} p(\gamma) p(\tilde{\gamma} |\gamma)\delta_{\gamma, \tilde{\gamma}_{\oplus_{2}}}
=\sum_{\gamma} \frac{1}{2} p(\tilde{\gamma}|\gamma)\delta_{\gamma, \tilde{\gamma}_{\oplus_{2}}}\nonumber\\
&=\sum_{\gamma, o_1} \frac{1}{2} p(\tilde{\gamma}, o_1 |\gamma)\delta_{\gamma, \tilde{\gamma}_{\oplus_{2}}}
=\sum_{\gamma, o_1} \frac{1}{2} p(o_1| \gamma) p(\tilde{\gamma}|\gamma, o_1)\delta_{\gamma, \tilde{\gamma}_{\oplus_{2}}}\nonumber\\
&=\frac{1}{8}\sum_{\gamma, o_1}p(\tilde{\gamma}|\gamma, o_1)\delta_{\gamma, \tilde{\gamma}_{\oplus_{2}}}  \\
&=\frac{1}{8}\sum_{\gamma, o_1}\sum_{y, i_2, o_2}p(\tilde{\gamma}, y, i_2, o_2 |\gamma,o_1)\delta_{\gamma, \tilde{\gamma}_{\oplus_{2}}}\nonumber\\
&=\frac{1}{8}\sum_{\gamma, o_1}\sum_{y, i_2, o_2}p(o_2|\gamma,o_1,\tilde{\gamma}, y, i_2 )p(\tilde{\gamma}, y, i_2 |\gamma,o_1)\delta_{\gamma, \tilde{\gamma}_{\oplus_{2}}}\nonumber\\
&=\frac{1}{8}\sum_{\gamma, o_1}\sum_{y, i_2, o_2}p(o_2|\gamma,o_1,\tilde{\gamma}, y, i_2 )p(\tilde{\gamma}, i_2 |\gamma,o_1, y)p(y| \gamma, o_1)\delta_{\gamma, \tilde{\gamma}_{\oplus_{2}}}~.\label{ns2}
\end{align}
\end{subequations}
We now analyze the different output cases, i.e. $o_1\in\{1,2,3,4\}$ separately:
\begin{itemize}
\item[] {\bf Case-I:} $o_1=1$. In this case we have $o_2=\gamma$. Thus for any NS correlation, we have
\begin{align}
p(\tilde{\gamma}|\gamma, o_1=1)\delta_{\gamma, \tilde{\gamma}_{\oplus_{2}}} = 1. \label{b1}
\end{align}
\item[] {\bf Case-II:} $o_1=2$. Here, we have $\tilde{\gamma} = (\hat{\Pi}_3(o_2)\oplus_3 b),~~ o_2 = i_2,~ p(y =0| \gamma, o_1=2) =0, ~\text{and}~ p(y =1| \gamma, o_1=2) =1$, which thus result in
\begin{align}
p(\tilde{\gamma}|\gamma, o_1=2)\delta_{\gamma, \tilde{\gamma}_{\oplus_{2}}} 
&= \sum_{i_2, o_2}p(o_2|\gamma,\tilde{\gamma}, i_2, o_1=2, y=1 )p(\tilde{\gamma}, i_2 |\gamma,o_1=2, y=1)\delta_{\gamma, \tilde{\gamma}_{\oplus_2}}\nonumber\\
&= \sum_{i_2, o_2} \delta_{o_2, i_2}\delta_{\gamma, (\Pi_3(o_2)\oplus_3 b)_{\oplus_2}}p(\Pi_3(o_2)\oplus_3 b, i_2 |\gamma,o_1=2, y=1)\nonumber\\
&= \sum_{i_2}\delta_{\gamma, (\Pi_3(i_2)\oplus_3 b)_{\oplus_2}}p(\Pi_3(i_2)\oplus_3 b, i_2 |\gamma,o_1=2, y=1)\nonumber\\
&= \sum_{i_2}\delta_{\gamma, (\Pi_3(i_2)\oplus_3 b)_{\oplus_2}}p(b, i_2 |\gamma,o_1=2, y=1)\nonumber\\
&= \sum_{i_2}\delta_{\gamma, (\Pi_3(i_2)\oplus_3 b)_{\oplus_2}}p(b, i_2 |\gamma, y=1)\nonumber\\
&= \sum_{a}\delta_{\gamma, (\Pi_3(a)\oplus_3 b)_{\oplus_2}}p(a, b |\gamma, y=1).\label{b2}
\end{align}
\item[] {\bf Case-III:} $o_1 = 3$. In this case, $\tilde{\gamma} = (o_2 \oplus_3 \Pi_3(b)),~~ o_2 = i_2\oplus_3 \gamma,~ p(y =0| \gamma, o_1=3) =1, ~\text{and}~ p(y =1| \gamma, o_1=3) =0$. thus we have 
\begin{align}
p(\tilde{\gamma}|\gamma, o_1=3)\delta_{\gamma, \tilde{\gamma}_{\oplus_{2}}} 
&= \sum_{i_2, o_2}p(o_2|\gamma,\tilde{\gamma}, i_2, o_1=3, y=0 )p(\tilde{\gamma}, i_2 |\gamma,o_1=3, y=0)\delta_{\gamma, \tilde{\gamma}_{\oplus_2}}\nonumber\\
&= \sum_{i_2, o_2} \delta_{o_2, i_2\oplus_3 \gamma}\delta_{\gamma, (o_2\oplus_3 \Pi_3(b))_{\oplus_2}}p(o_2\oplus_3 \Pi_3(b), i_2 |\gamma,o_1=3, y=0)\nonumber\\
&= \sum_{i_2} \delta_{\gamma, ( i_2\oplus_3 \gamma\oplus_3 \Pi_3(b))_{\oplus_2}}p(i_2\oplus_3 \gamma\oplus_3 \Pi_3(b), i_2 |\gamma,o_1=3, y=0)\nonumber\\
&= \sum_{i_2}\delta_{\gamma, (i_2\oplus_3 \gamma\oplus_3 \Pi_3(b))_{\oplus_2}}p(b, i_2 |\gamma,o_1=3, y=0)\nonumber\\
&= \sum_{i_2}\delta_{\gamma, (i_2\oplus_3 \gamma\oplus_3 \Pi_3(b))_{\oplus_2}}p(b, i_2 |\gamma, y=0)\nonumber\\
&= \sum_{a}\delta_{\gamma, (a\oplus_3 \gamma\oplus_3 \Pi_3(b))_{\oplus_2}}p(a, b |\gamma, y=0). \label{b3}
\end{align}
\item[] {\bf Case-IV:} $o_1 = 4$. Here, $\tilde{\gamma} = (o_2 \oplus_3 b),~~ o_2 = i_2\oplus_3 \gamma,~ p(y =0| \gamma, o_1=4) =1, ~\text{and}~ p(y =1| \gamma, o_1=4) =0$, implying
\begin{align}
p(\tilde{\gamma}|\gamma, o_1=4)\delta_{\gamma, \tilde{\gamma}_{\oplus_{2}}} 
=& \sum_{i_2, o_2}p(o_2|\gamma,\tilde{\gamma}, i_2, o_1=4, y=0 )p(\tilde{\gamma}, i_2 |\gamma,o_1=4, y=0)\delta_{\gamma, \tilde{\gamma}_{\oplus_2}}\nonumber\\
=& \sum_{i_2, o_2} \delta_{o_2, \Pi(i_2)\oplus_3 \gamma}\delta_{\gamma, (o_2\oplus_3 b)_{\oplus_2}}p(o_2\oplus_3 b, i_2 |\gamma,o_1=4, y=0)\nonumber\\\nonumber
=& \sum_{i_2} \delta_{\gamma, ( \Pi(i_2)\oplus_3 \gamma\oplus_3 b)_{\oplus_2}}p(\Pi(i_2)\oplus_3 \gamma\oplus_3 b, i_2 |\gamma,o_1=4, y=0)\\\nonumber
=& \sum_{i_2}\delta_{\gamma, (\Pi(i_2)\oplus_3 \gamma\oplus_3 b)_{\oplus_2}}p(b, i_2 |\gamma,o_1=4, y=0)\\
=& \sum_{i_2}\delta_{\gamma, (\Pi(i_2)\oplus_3 \gamma\oplus_3 b)_{\oplus_2}}p(b, i_2 |\gamma, y=0)\nonumber\\
=&\sum_{a}\delta_{\gamma, (\Pi(a)\oplus_3 \gamma\oplus_3 b)_{\oplus_2}}p(a, b |\gamma, y=0).\label{b4}
\end{align}
\end{itemize}
Substituting Eqs.(\ref{b1}-\ref{b4}) in Eq.(\ref{ns2}), for the correlation in Eq.(\ref{mec}), we obtain
\begin{align}
\$_{\phi^+}(\mathcal{N}_3)
=& \frac{1}{8}\sum_{\gamma}\left[p(\tilde{\gamma}|\gamma, o_1=1)\delta_{\gamma, \tilde{\gamma}_{\oplus_{2}}} + p(\tilde{\gamma}|\gamma, o_1=2)\delta_{\gamma, \tilde{\gamma}_{\oplus_{2}}} +p(\tilde{\gamma}|\gamma, o_1=3)\delta_{\gamma, \tilde{\gamma}_{\oplus_{2}}} \right.\nonumber\\
&\left.\hspace{3cm}+p(\tilde{\gamma}|\gamma, o_1=4)\delta_{\gamma, \tilde{\gamma}_{\oplus_{2}}}\right]\nonumber\\ 
=& \frac{1}{8}\left[2 +\sum_{\gamma, a}\delta_{\gamma, (\Pi_3(a)\oplus_3 b)_{\oplus_2}}p(a, b |\gamma, y=1)+\delta_{\gamma, (a\oplus_3 \gamma\oplus_3 \Pi_3(b))_{\oplus_2}}p(a, b |\gamma, y=0)\right.\nonumber\\
&\left.\hspace{4cm}+\delta_{\gamma, (\Pi(a)\oplus_3 \gamma\oplus_3 b)_{\oplus_2}}p(a, b |\gamma, y=0)\right]\nonumber\\
=& \frac{1}{4}\left[1+\frac{1}{36} \text{C}^2\left(\frac{\pi}{4}\right)+\frac{1}{18} \text{C}^2\left(\frac{5\pi}{12}\right)+ \frac{1}{6} \text{C}^2\left(\frac{\pi}{12}\right)\right]\approx~0.9008~~.
\end{align}
\section{Advantage of $2$-$m$-$2$ extremal nonlocal correlations in zero-error communication}\label{appendix:c}
\noindent {\bf $m_1m_222$ Bell scenario.--} In this case we have \( |\mathcal{X}| = m_1 \geq 2 \), \( |\mathcal{Y}| = m_2 \geq 2 \), and \( |\mathcal{A}| = |\mathcal{B}| = 2 \), i.e., $x\in \{0,1,\cdots, m_1-1\}$, $y\in\{0,1,\cdots,m_2-1\}$ and outputs $a,b \in \{0,1\}$. Let us denote the extremal nonlocal correlations as \(\mathcal{R}_{m_1 m_2}\). This set inherently includes the extended versions of the extreme correlations of lower input dimensions, \(\mathcal{R}_{m_1' m_2'}\), where \(m_1' < m_1\) and \(m_2' < m_2\). For such correlations, individual parties produce constant outputs for all \(m_1 > m_1'\) and \(m_2 > m_2'\). The extremal correlations $\mathcal{\tilde{R}}_{m_1m_2}\equiv \{r_{m_1m_2}(a, b | x, y)\}$ which are not extended versions of lower input correlations are characterized by Jones {\it et.al.} \cite{Jones2005} in the following form
\begin{align}
&r_{m_1m_2}(a, b | x, y)=\frac{1}{2},~~\text{if}~~a \oplus_2 b = \delta_{x,1}\delta_{y,1} + \sum_{(i,j)\in Q}\delta_{x,i}\delta_{y,j}\\
&\forall~Q\subseteq\{1,2,\cdots, m_1-1\}\times \{1,2,\cdots,m_2-1\}-\{(1,1)\}.\nonumber
\end{align}
In this work, we will only consider a particular class of the extremal NS correlations \( \mathcal{\tilde{R}}_m \equiv \{r_m(a, b | x, y)\} \) given by: \\
\begin{align}
r_m(a, b | x, y) := 
\begin{cases} 
\frac{1}{2}, ~~\text{if } a \oplus_2 b = \sum_{i=1}^{m-1} \delta_{q,i} \cdot \delta_{\alpha,i}~,~\text{with}~m:=\min\{m_1,m_2\}; \\     
0,~~\text{otherwise}.
\end{cases}
\end{align}
{\bf The channel $\mathcal{M}_{m}$.--} For $m\ge2$, the channel's inputs are the tuple $(i_{1},i_{2}) \in \{0,1,\cdots,m-1\}\times\{0,1\}$, while outputs are the tuple $(o_{1},o_{2})\in \{1,2,\cdots,m(m-1)+1\}\times\{0,1,\cdots, m-1\}$ (see Fig.\ref{fig5}). For any input, all the possible outcomes of $o_{1}$ are uniformly at random, and the outcome $o_2$ depends on $o_1$ in the following manner:
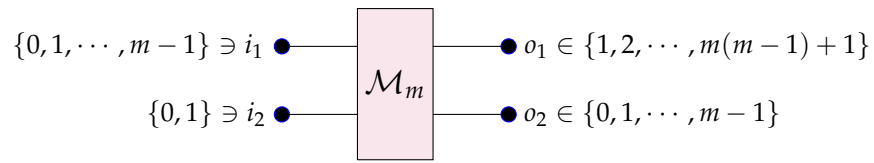
\begin{figure}[h!]
\centering
\begin{tikzpicture}
\draw[fill=purple!10](-2,0)--(-2,2)--(-1,2)--(-1,0)--(-2,0);
\draw[thin,black](-3,1.5)--(-2,1.5);
\draw[thin,black](-3,.6)--(-2,.6);
\draw[thin,black](-1,1.5)--(0,1.5);
\draw[thin,black](-1,.6)--(0,.6);
\draw[blue, fill=black] (-3,1.5) circle (.1);
\draw[blue, fill=black] (-3,.6) circle (.1);
\draw[blue, fill=black] (0,1.5) circle (.1);
\draw[blue, fill=black] (0,.6) circle (.1);
\node[] at (-1.5,1){\large{$\mathcal{M}_m$}};
\node[] at (-4.9,1.5){$\{0,1,\cdots,m-1\}\ni i_1$};
\node[] at (-4.,.6){$\{0,1\}\ni i_2$};
\node[] at (2.5,1.5){$o_1\in\{1,2,\cdots, m(m-1)+1\}$};
\node[] at (1.9,.6){$o_2\in\{0,1,\cdots,m-1\}$};
\end{tikzpicture}
\caption{Diagrammatical illustration of the channel $\mathcal{M}_m$. For $m\ge2$, $(i_{1},i_{2}) \in \{0,1,\cdots,m-1\}\times\{0,1\}$, whereas $(o_{1},o_{2})\in \{1,2,\cdots,m(m-1)+1\}\times\{0,1,\cdots, m-1\}$.}\label{fig5}
~\vspace{-.5cm}
\end{figure}
\begin{align}
o_2:=\begin{cases}
i_1,~&\text{if } o_1 = 1;\\
i_1 \oplus_{3} \Pi^{o_1-2}_{m}(i_2),~&\text{if } o_1 = 2;\\
~~~~\vdots~&~~~~\vdots\\
i_1 \oplus_{3} \Pi^{o_1-2}_{m}(i_2),~&\text{if } o_1 = m;\\
i_1 \oplus_{3} \Pi^{o_1-(m+1)}_{m}(i_2 \oplus_2 \delta_{i_{1},1}),~&\text{if } o_1 = m+1;\\
~~~~\vdots~&~~~~\vdots\\
i_1 \oplus_{3} \Pi^{o_1-(m+1)}_{m}(i_2 \oplus_2 \delta_{i_{1},1}),~&\text{if } o_1 = 2m-1;\\
~~~~\vdots~&~~~~\vdots\\
~~~~\vdots~&~~~~\vdots\\
i_1 \oplus_{3} \Pi^{o_1-((m-1)^2+2)}_{m}(i_2 \oplus_2 \delta_{i_{1},m-1}),~&\text{if } o_1 = (m-1)^2+2;\\
~~~~\vdots~&~~~~\vdots\\
i_1 \oplus_{3} \Pi^{o_1-((m-1)^2+2)}_{m}(i_2 \oplus_2 \delta_{i_{1},m-1}),~&\text{if } o_1 = m(m-1)+1;
\end{cases}
\end{align}
where $\Pi^{o_1-k}_{3}(0):=0$ and for $i_2\neq 0,~\Pi_3^{o_1-k}(i_2):= [(i_{2}-1)\oplus_{2}(o_{1}-k)]+1$. For the channel $\mathcal{M}_m$,  $o_1$ takes values from the set $\{1,2,\cdots,m(m-1)+1\}$. For $o_1=1$, we define the set
\begin{align}
\mathcal{K}:=\{(o_1,o_2)~|~o_1=1~\&~o_2\in\{0,1,\cdots,m-1\}\}.
\end{align}
We partition the remaining outputs of the channel into $m$ sets $\mathcal{S}^j$, with $j\in\{0,1,\cdots,m-1\}$, each containing $m(m-1)$ outputs and defined as
\begin{align}
\mathcal{S}^j&:=\left\{(o_1,o_2)~|~o_1\in\{(m-1)j+2,(m-1)j+3,\cdots,(m-1)j+m\}\right.~\nonumber\\
&\hspace{5cm}\left.\&~o_2\in\{0,1,\cdots,m-1\}\right\}.
\end{align}
The channel $\mathcal{M}_m$ has the following features:
\begin{itemize}
\item[(P0)] Each of the input $(i_1,i_2)\in\{0,1,\cdots,m-1\}\times\{0,1\}$ results in a subset of outputs $\mathcal{K}_{i_1,i_2}:=\{(o_1,o_2)~|~o_1=1~\&~o_2=i_1\}$ in $\mathcal{K}$. Notably, 
\begin{subequations}
\begin{align}
&\bigcup_{i_1}\mathcal{K}_{(i_1,i_2)}=\mathcal{K}~,\text{for}~i_2=0,1;\\
&\mathcal{K}_{(i_1,i_2)}\bigcap\mathcal{K}_{(i'_1,i'_2)}=\emptyset~,\text{whenever}~i_1\neq i'_1~;\\
&\mathcal{K}_{(i_1,0)}\bigcap\mathcal{K}_{(i_1,1)}\neq\emptyset.
\end{align}
\end{subequations}  
\item[(P1)] For all $i_i\in\{0,1,\cdots,m-1\}$, the input $(i_1,i_2=0)$ results in a subset of outputs $\mathcal{S}^0_{(i_1,0)}\equiv\{(o_1,o_2=i_1)~|~o_1\in\{2,3,\cdots,m\}\}$ in $\mathcal{S}^0$. Similarly, input $(i_1,i_2=1)$ results in a subset of outputs $\mathcal{S}^0_{(i_1,1)}\equiv\{(o_1,o_2=i_1\oplus_m\Pi_m^{o_1-2}(i_2))~|~o_1\in\{2,3,\cdots,m\}$ in $\mathcal{S}^0$. Notably, 
\begin{subequations}
\begin{align}
&\bigcup_{i_1}\mathcal{S}^0_{(i_1,i_2)}=\mathcal{S}^0~,\text{for}~i_2=0,1;\\
&\mathcal{S}^0_{(i_1,i_2)}\bigcap\mathcal{S}^0_{(i'_1,i_2)}=\emptyset~,\text{whenever}~i_1\neq i'_1~;\\
&\mathcal{S}^0_{(i_1,0)}\bigcap\mathcal{S}^0_{(i'_1,1)}\begin{cases}
=\emptyset~,\text{whenever}~i_1=i'_1~.\\
\neq\emptyset~,\text{whenever}~i_1\neq i'_1~.
\end{cases}
\end{align}
\end{subequations}
\item[(P2)] Consider the case $j\neq 0$, then for all $i_i\in\{0,1,\cdots,m-1\}$, the input $(i_1,i_2)$ results in one of the following subsets of outputs:
\begin{subequations}
\begin{align}
&\mathcal{S}^{j}_{(i_1,0)}\equiv\{(o_1,o_2)~|~o_1\in\{(m-1)j+2,\cdots,(m-1)j+m\} ~\&~ o_2 = i_1\} \subseteq \mathcal{S}^j;\\
&\mathcal{S}^j_{(j,0)}\equiv\{(o_1,o_2)~|~o_1\in\{(m-1)j+2,\cdots,(m-1)j+m\} \nonumber\\
&\hspace{5cm}~\&~ o_2=i_1\oplus_m\Pi_m^{o_1-((m-1)j+2)}(1)\} \subseteq \mathcal{S}^j;\\
&\mathcal{S}^j_{(i_1,1)}\equiv\{(o_1,o_2)~|~o_1\in\{(m-1)j+2,\cdots,(m-1)j+m\} \nonumber\\
&\hspace{5cm}~\&~ o_2 = i_1\oplus_m\Pi_m^{o_1-((m-1)j+2)}(1)\} \subseteq \mathcal{S}^j;\\
&\mathcal{S}^{j}_{(j,1)}\equiv\{(o_1,o_2)~|~o_1\in\{(m-1)j+2,\cdots,(m-1)j+m\} ~\&~ o_2 = j\} \subseteq \mathcal{S}^j.
\end{align}
\end{subequations}
Notably, 
\begin{subequations}
\begin{align}
&\mathcal{S}^j_{(i_1\neq j,i_2)}\bigcap\mathcal{S}^j_{(j,i_2)}\neq\emptyset~,\forall~ i_1, i_2, j;\label{eqp2}
\end{align}    
\end{subequations}
\end{itemize}
All the above properties can be explicitly verified for the simple case of $m=3$. The details of the channel $\mathcal{M}_m$ are provided in Table \ref{tab6}.\\
\begin{table}[h!]
\centering
\begin{tabular}{c||c|c|c|c|c|c|}
\hline
$\text{Output} \backslash \text{Input}$~~~ &~~~ $(0,0)$~~~ &~~~ $(0,1)$~~~ & ~~~$(1,0)$~~~ &~~~ $(1,1)$~~~ &~~~ $(2,0)$~~~ & ~~~$(2,1)$~~~ \\ \hline\hline
$(1,0)$ & $1/7$ & $1/7$ & $0$   & $0$   & $0$   & $0$   \\\hline
$(1,1)$ & $0$   & $0$   & $1/7$ & $1/7$ & $0$   & $0$   \\\hline
$(1,2)$ & $0$   & $0$   & $0$   & $0$   & $1/7$ & $1/7 $\\\hline
\rowcolor{green!20}$(2,0)$ & $1/7$ & $0$   & $0$   & $0$   & $0$   & $1/7$ \\\hline
\rowcolor{yellow!20}$(2,1)$ & $0$   & $1/7$ & $1/7$ & $0$   & $0$   & $0$   \\\hline
$(2,2)$ & $0$   & $0$   & $0$   & $1/7$ & $1/7$ & $0$   \\\hline
\rowcolor{blue!20}$(3,0)$ & $1/7$ & $0$   & $0$   & $1/7$ & $0$   & $0$   \\\hline
$(3,1)$ & $0$   & $0$   & $1/7$ & $0$   & $0$   & $1/7$ \\\hline
\rowcolor{green!50!blue!50}$(3,2)$ & $0$   & $1/7$ & $0$   & $0$   & $1/7$ & $0$   \\\hline
\rowcolor{green!20}$(4,0)$ & $1/7$   & $0$ & $0$   & $0$   & $0$ & $1/7$   \\\hline
$(4,1)$ & $0$   & $1/7$ & $0$  & $1/7$ & $0$   & $0$  \\\hline
\rowcolor{purple!20}$(4,2)$ & $0$   & $0$  & $1/7$ &$0$   & $1/7$ & $0$   \\\hline
$(5,0)$ & $1/7$ & $0$   & $1/7$ & $0$   & $0$  & $0$   \\\hline
\rowcolor{red!30!blue!40}$(5,1)$ & $0$   & $0$   & $0$   & $1/7$ & $0$   & $1/7$ \\\hline
\rowcolor{green!50!blue!50}$(5,2)$ &$0$  & $1/7$ & $0$   & $0$  & $1/7$ &$0$   \\\hline
$(6,0)$ & $1/7$ & $0$   & $0$   & $0$   & $1/7$ &$0$   \\\hline
\rowcolor{yellow!20}$(6,1)$ & $0$   & $1/7$ & $1/7$ & $0$   & $0$   & $0$   \\\hline
\rowcolor{red!30!blue!40}$(6,2)$ & $0$  & $0$   & $0$   & $1/7$ & $0$ & $1/7$   \\\hline
\rowcolor{blue!20}$(7,0)$ & $1/7$ & $0$   & $0$   & $1/7$ & $0$   & $0$   \\\hline
\rowcolor{purple!20}$(7,1)$ & $0$   & $0$   & $1/7$ & $0$   &  $1/7$ & $0$  \\\hline
$(7,2)$ & $0$   & $1/7$ & $0$   & $0$   & $0$   & $1/7$ \\\hline
\end{tabular}
\caption{The stochastic matrix \(\mathbb{M}_3 \equiv \{\text{Pr}(o_1,o_2|i_1,i_2)\}\) represents the channel \(\mathcal{M}_3\), where each output induces confusability between exactly two inputs. For example, the output \((1,0)\) is associated with the inputs \((0,0)\) and \((0,1)\), causing confusability. In some cases, input pairs are confused at more than one output. For instance, the inputs \((0,0)\) and \((2,1)\) are confusable for the outputs \((2,0)\) and \((4,0)\). Such instances are highlighted using distinct colors for clarity.}
\label{tab6}
\end{table}\\
The total number of edges in a fully connected graph with $2m$ nodes is $^{2m}C_{2} = m(2m-1)$.
\begin{itemize}
\item[(C0)] Property (P0) leads to confusability between all input pairs $(i_1, 0)$ and $(i_1,1)$ and gives $m$ unique edges of the graph.
\item[(C1)] Outputs in $S^0$ confuse inputs pairs $(i_1, 0)$ with input pairs $(i_1',1)$ for all $i_1\neq i_1'$ and results in $m(m-1)$ number of edges.
\item[(C2)] C0 and C1, confuses input pairs $(i_1, 0)$ and $(i_1',1)$ for all $i_1$ and $i_1'$. Eq.(\ref{eqp2}), for a particular value of $j$, confuses input pairs $(j,i_2)$ and $(i_1 \neq j, i_2)$ for all $i_2$. Since this is true for all $j$, all $(i_1, i_2)$ become confusable with $(i_1'\neq i_1, i_2)$ for all $i_2$.
\end{itemize}
This shows that any input of the channel is confusable with any other and hence
the confusability graph $\mathbb{G}(\mathcal{M}_{m})$ is fully connected, which further implies the $C^0_{SR}(\mathcal{M}_{m}) = 0$.

\section{Advantage of Entanglement Assistance in $\mathcal{M}_3$ Channel [Proposition 2]}\label{appendix:d}
\noindent The NS correlation $R_3$ utilized in Theorem 3 reads as
\begin{align}
\tilde{R}_3\equiv\begin{array}{c||c|c|c|c|c|c|c|c|c|}
\hline
a,b\backslash x,y & 0,0 & 0,1& 0,2& 1,0& 1,1& 1,2& 2,0& 2,1& 2,2 \\\hline\hline
0,0& 1/2& 1/2 & 1/2 & 1/2 & 0 & 1/2 & 1/2 & 1/2 &0\\\hline
0,1& 0& 0 & 0 & 0 & 1/2 & 0 & 0 & 0 &1/2\\\hline
1,0& 0& 0 & 0 & 0 & 1/2 & 0 & 0 & 0 &1/2\\\hline
1,1& 1/2& 1/2 & 1/2 & 1/2 & 0 & 1/2 & 1/2 & 1/2 &0\\\hline
\end{array}~.  
\end{align}
Nonlocality of this correlation is established by the violation of relabeling of three-input-two-output Bell inequality $(\mathbb{I}_{3322})$ proposed by Collins and Gisin \cite{Collins2004}. The correlation $\tilde{R}_3$ violates this $\mathbb{I}_{3322}$ inequality to its algebraic maximum. The authors in \cite{Collins2004} have also analyzed the optimal violation of this inequality by two-qubit maximally entangled state $\ket{\psi^-}:=\frac{1}{\sqrt{2}}(\ket{01}-\ket{10})\in\mathbb{C}^2\otimes\mathbb{C}^2$. The measurement choice of individual parties lies in the $X-Z$ plane of the Bloch sphere and denoted by the single angle they make with the $Z$ axis, $A_0=0, A_1=\frac{\pi}{3}, A_2=\frac{2\pi}{3}, B_0=\frac{4\pi}{3}, B_1=\frac{2\pi}{3}, B_2=\pi$. The correlation obtained from these measurements is given by
\begin{align}
R^{\psi^-}_3\equiv\begin{array}{c||c|c|c|c|c|c|c|c|c|}
\hline
a,b\backslash x,y & 0,0 & 0,1& 0,2& 1,0& 1,1& 1,2& 2,0& 2,1& 2,2 \\\hline\hline
0,0& 3/8& 3/8 & 1/2 & 1/2 & 1/8 & 3/8 & 3/8 & 1/2 &1/8\\\hline
0,1& 1/8& 1/8 & 0 & 0 & 3/8 & 1/8 & 1/8 & 0 &3/8\\\hline
1,0& 1/8& 1/8 & 0 & 0 & 3/8 & 1/8 & 1/8 & 0 &3/8\\\hline
1,1& 3/8& 3/8 & 1/2 & 1/2 & 1/8 & 3/8 & 3/8 & 1/2 &1/8\\\hline
\end{array}~\label{ent}.
\end{align}\\
While sending the information $\gamma\in\{0,1,2\}$ through the channel $\mathcal{M}_3$ in assistance with a generic NS correlation by employing the strategy of Table \ref{tab4}, Bob's guess $\tilde{\gamma}$ belongs to the set $\{0,1,2\}$. Since $\gamma$ is uniformly random over $\{0,1,2\}$, Bob can choose to answer $\tilde{\gamma}$ as $\gamma$ . Denoting $x_{\oplus_{2}}:= x ~\text{mod}~ 2$, the success probability of the protocol reads as,\\
\begin{subequations}
\begin{align}
\$_{NS}(\mathcal{M}_3)&=p(\tilde{\gamma} = \gamma ,\gamma)
=\sum_{\gamma} p(\gamma) p(\tilde{\gamma} |\gamma)\delta_{\gamma, \tilde{\gamma}}
=\sum_{\gamma} \frac{1}{3} p(\tilde{\gamma}|\gamma)\delta_{\gamma, \tilde{\gamma}}\nonumber\\
&=\sum_{\gamma, o_1} \frac{1}{3} p(\tilde{\gamma}, o_1 |\gamma)\delta_{\gamma, \tilde{\gamma}}
=\sum_{\gamma, o_1} \frac{1}{3} p(o_1| \gamma) p(\tilde{\gamma}|\gamma, o_1)\delta_{\gamma, \tilde{\gamma}}\nonumber\\
&=\frac{1}{21}\sum_{\gamma, o_1}p(\tilde{\gamma}|\gamma, o_1)\delta_{\gamma, \tilde{\gamma}}  \\
&=\frac{1}{21}\sum_{\gamma, o_1}\sum_{y, i_2, o_2}p(\tilde{\gamma}, y, i_2, o_2 |\gamma,o_1)\delta_{\gamma, \tilde{\gamma}}=\frac{1}{8}\sum_{\gamma, o_1}\sum_{y, i_2, o_2}p(o_2|\gamma,o_1,\tilde{\gamma}, y, i_2 )p(\tilde{\gamma}, y, i_2 |\gamma,o_1)\delta_{\gamma, \tilde{\gamma}}\nonumber\\
&=\frac{1}{21}\sum_{\gamma, o_1}\sum_{y, i_2, o_2}p(o_2|\gamma,o_1,\tilde{\gamma}, y, i_2 )p(\tilde{\gamma}, i_2 |\gamma,o_1, y)p(y| \gamma, o_1)\delta_{\gamma, \tilde{\gamma}}~.\label{ns3}
\end{align}
\end{subequations}
We now analyze the different output cases, i.e. $o_1\in\{1,2,3,4,5,6,7\}$ separately:
\begin{itemize}
\item[] {\bf Case-I:} $o_1=1$. In this case, we have $o_2=\gamma$. Thus for any NS correlation, we have
\begin{align}
p(\tilde{\gamma}|\gamma, o_1=1)\delta_{\gamma, \tilde{\gamma}} = 1. \label{d1}
\end{align}
\item[] {\bf Case-II:} $o_1=2$. Here, we have $~o_2 = \gamma\oplus_3\Pi^{0}_{3}(i_2),~\tilde{\gamma} = (\hat{\Pi}_3(\Pi^{0}_{3}(b))\oplus_3 o_2),~ p(y =0| \gamma, o_1=2) =1, $\\$ p(y =1| \gamma, o_1=2) =0 ~\text{and}~,p(y =2| \gamma, o_1=2) =0$, which thus result in
\begin{align}
&p(\tilde{\gamma}|\gamma, o_1=2)\delta_{\gamma, \tilde{\gamma}} 
= \sum_{i_2, o_2}p(o_2|\gamma,\tilde{\gamma}, i_2, o_1=2, y=0 )p(\tilde{\gamma}, i_2 |\gamma,o_1=2, y=0)\delta_{\gamma, \tilde{\gamma}}\nonumber\\
&= \sum_{i_2, o_2} \delta_{o_2, (\gamma\oplus_3\Pi^{0}_{3}(i_2))}\delta_{\gamma, (\hat{\Pi}_3(\Pi^{0}_{3}(b))\oplus_3 o_2)}p\left(\hat{\Pi}_3(\Pi^{0}_{3}(b))\oplus_3 (\gamma\oplus_3\Pi^{0}_{3}(i_2)), i_2 |\gamma,o_1=2, y=0\right)\nonumber\\
&= \sum_{i_2}\delta_{\gamma, \left(\hat{\Pi}_3\left(\Pi^{0}_{3}(b)\right)\oplus_3 \left(\gamma\oplus_3\Pi^{0}_{3}(i_2)\right)\right)}p\left(\hat{\Pi}_3(\Pi^{0}_{3}(b))\oplus_3 (\gamma\oplus_3\Pi^{0}_{3}(i_2)), i_2 |\gamma,o_1=2, y=0\right)\nonumber\\
&= \sum_{i_2}\delta_{\gamma, \left(\hat{\Pi}_3\left(\Pi^{0}_{3}(b)\right)\oplus_3 \left(\gamma\oplus_3\Pi^{0}_{3}(i_2)\right)\right)}p(b, i_2 |\gamma,o_1=2, y=0)\nonumber\\
&= \sum_{i_2}\delta_{\gamma, \left(\hat{\Pi}_3\left(\Pi^{0}_{3}(b)\right)\oplus_3 \left(\gamma\oplus_3\Pi^{0}_{3}(i_2)\right)\right)}p(b, i_2 |\gamma, y=0)\nonumber\\
&= \sum_{a}\delta_{\gamma, \left(\hat{\Pi}_3\left(\Pi^{0}_{3}(b)\right)\oplus_3 \left(\gamma\oplus_3\Pi^{0}_{3}(a)\right)\right)}p(a, b |\gamma, y=0).\label{d2}
\end{align}
\item[] {\bf Case-III:} $o_1 = 3$. In this case, $o_2 = \gamma\oplus_3\Pi^{1}_{3}(i_2),~~\tilde{\gamma} = (\hat{\Pi}_3(\Pi^{1}_{3}(b))\oplus_3 o_2),~ p(y =0| \gamma, o_1=3) =1, $\\$ p(y =1| \gamma, o_1=2) =0 ~\text{and}~,p(y =2| \gamma, o_1=2) =0$. thus we have 
\begin{align}
&p(\tilde{\gamma}|\gamma, o_1=3)\delta_{\gamma, \tilde{\gamma}} 
= \sum_{i_2, o_2}p(o_2|\gamma,\tilde{\gamma}, i_2, o_1=3, y=0 )p(\tilde{\gamma}, i_2 |\gamma,o_1=3, y=0)\delta_{\gamma, \tilde{\gamma}}\nonumber\\
&= \sum_{i_2, o_2} \delta_{o_2, \gamma\oplus_3\Pi^{1}_{3}(i_2)}\delta_{\gamma, (\hat{\Pi}_3(\Pi^{1}_{3}(b))\oplus_3 o_2)}p(\hat{\Pi}_3(\Pi^{1}_{3}(b))\oplus_3 o_2, i_2 |\gamma,o_1=3, y=0)\nonumber\\
&= \sum_{i_2} \delta_{\gamma, (\hat{\Pi}_3(\Pi^{1}_{3}(b))\oplus_3 \gamma\oplus_3\Pi^{1}_{3}(i_2))}p((\hat{\Pi}_3(\Pi^{1}_{3}(b))\oplus_3 \gamma\oplus_3\Pi^{1}_{3}(i_2)), i_2 |\gamma,o_1=3, y=0)\nonumber\\
&= \sum_{i_2}\delta_{\gamma, (\hat{\Pi}_3(\Pi^{1}_{3}(b))\oplus_3 \gamma\oplus_3\Pi^{1}_{3}(i_2))}p(b, i_2 |\gamma,o_1=3, y=0)\nonumber\\
&= \sum_{i_2}\delta_{\gamma, (\hat{\Pi}_3(\Pi^{1}_{3}(b))\oplus_3 \gamma\oplus_3\Pi^{1}_{3}(i_2))}p(b, i_2 |\gamma, y=0)\nonumber\\
&= \sum_{a}\delta_{\gamma, (\hat{\Pi}_3(\Pi^{1}_{3}(b))\oplus_3 \gamma\oplus_3\Pi^{1}_{3}(a))}p(a, b |\gamma, y=0). \label{d3}
\end{align}
\item[] {\bf Case-IV:} $o_1 = 4$. Here, $o_2 = \gamma\oplus_3\Pi^{0}_{3}(i_2\oplus_2 \delta_{\gamma,1}) ,~~\tilde{\gamma} = \left(o_2 \oplus_3 \hat{\Pi_{3}}(\Pi^{0}_{3}(b))\right),~~ p(y =0| \gamma, o_1=4) =0,$\\$ p(y =1| \gamma, o_1=4) =1~\text{and}~,p(y =2| \gamma, o_1=4) =0$, implying
\begin{align}
&p(\tilde{\gamma}|\gamma, o_1=4)\delta_{\gamma, \tilde{\gamma}_{\oplus_{2}}} 
= \sum_{i_2, o_2}p(o_2|\gamma,\tilde{\gamma}, i_2, o_1=4, y=0 )p(\tilde{\gamma}, i_2 |\gamma,o_1=4, y=0)\delta_{\gamma, \tilde{\gamma}}\nonumber\\
=& \sum_{i_2, o_2} \delta_{o_2, \left(\gamma\oplus_3\Pi^{0}_{3}(i_2\oplus_2 \delta_{\gamma,1})\right)}\delta_{\gamma, \left(o_2 \oplus_3 \hat{\Pi_{3}}(\Pi^{0}_{3}(b))\right)}p(o_2 \oplus_3 \hat{\Pi_{3}}(\Pi^{0}_{3}(b)), i_2 |\gamma,o_1=4, y=1)\nonumber\\\nonumber
=& \sum_{i_2} \delta_{\gamma, \left(\gamma\oplus_3\Pi^{0}_{3}(i_2\oplus_2 \delta_{\gamma,1}) \oplus_3 \hat{\Pi_{3}}(\Pi^{0}_{3}(b))\right)}p(\gamma\oplus_3\Pi^{0}_{3}(i_2\oplus_2 \delta_{\gamma,1}) \oplus_3 \hat{\Pi_{3}}(\Pi^{0}_{3}(b)), i_2 |\gamma,o_1=4, y=1)\\\nonumber
=& \sum_{i_2}\delta_{\gamma, \left(\gamma\oplus_3\Pi^{0}_{3}(i_2\oplus_2 \delta_{\gamma,1}) \oplus_3 \hat{\Pi_{3}}(\Pi^{0}_{3}(b))\right)}p(b, i_2 |\gamma,o_1=4, y=1)\\
=& \sum_{i_2}\delta_{\gamma, \left(\gamma\oplus_3\Pi^{0}_{3}(i_2\oplus_2 \delta_{\gamma,1}) \oplus_3 \hat{\Pi_{3}}(\Pi^{0}_{3}(b))\right)}p(b, i_2 |\gamma, y=1)\nonumber\\
=&\sum_{a}\delta_{\gamma, \left(\gamma\oplus_3\Pi^{0}_{3}(a\oplus_2 \delta_{\gamma,1}) \oplus_3 \hat{\Pi_{3}}(\Pi^{0}_{3}(b))\right)}p(a, b |\gamma, y=1).\label{d4}
\end{align}
\item[] {\bf Case-V:} $o_1 = 5$. Here, $o_2 = \gamma\oplus_3\Pi^{1}_{3}(i_2\oplus_2 \delta_{\gamma,1}) ,~~\tilde{\gamma} = \left(o_2 \oplus_3 \hat{\Pi_{3}}(\Pi^{1}_{3}(b))\right),~~ p(y =0| \gamma, o_1=5) =0,$\\$ p(y =1| \gamma, o_1=5) =1~\text{and}~,p(y =2| \gamma, o_1=5) =0$, implying
\begin{align}
&p(\tilde{\gamma}|\gamma, o_1=5)\delta_{\gamma, \tilde{\gamma}_{\oplus_{2}}} 
= \sum_{i_2, o_2}p(o_2|\gamma,\tilde{\gamma}, i_2, o_1=5, y=0 )p(\tilde{\gamma}, i_2 |\gamma,o_1=5, y=0)\delta_{\gamma, \tilde{\gamma}}\nonumber\\
=& \sum_{i_2, o_2} \delta_{o_2, \left(\gamma\oplus_3\Pi^{1}_{3}(i_2\oplus_2 \delta_{\gamma,1})\right)}\delta_{\gamma, \left(o_2 \oplus_3 \hat{\Pi_{3}}(\Pi^{1}_{3}(b))\right)}p(o_2 \oplus_3 \hat{\Pi_{3}}(\Pi^{1}_{3}(b)), i_2 |\gamma,o_1=5, y=1)\nonumber\\\nonumber
=& \sum_{i_2} \delta_{\gamma, \left(\gamma\oplus_3\Pi^{1}_{3}(i_2\oplus_2 \delta_{\gamma,1}) \oplus_3 \hat{\Pi_{3}}(\Pi^{1}_{3}(b))\right)}p(\gamma\oplus_3\Pi^{1}_{3}(i_2\oplus_2 \delta_{\gamma,1}) \oplus_3 \hat{\Pi_{3}}(\Pi^{1}_{3}(b)), i_2 |\gamma,o_1=5, y=1)\\\nonumber
=& \sum_{i_2}\delta_{\gamma, \left(\gamma\oplus_3\Pi^{1}_{3}(i_2\oplus_2 \delta_{\gamma,1}) \oplus_3 \hat{\Pi_{3}}(\Pi^{1}_{3}(b))\right)}p(b, i_2 |\gamma,o_1=5, y=1)\\
=& \sum_{i_2}\delta_{\gamma, \left(\gamma\oplus_3\Pi^{1}_{3}(i_2\oplus_2 \delta_{\gamma,1}) \oplus_3 \hat{\Pi_{3}}(\Pi^{1}_{3}(b))\right)}p(b, i_2 |\gamma, y=1)\nonumber\\
=&\sum_{a}\delta_{\gamma, \left(\gamma\oplus_3\Pi^{1}_{3}(a\oplus_2 \delta_{\gamma,1}) \oplus_3 \hat{\Pi_{3}}(\Pi^{1}_{3}(b))\right)}p(a, b |\gamma, y=1).\label{d5}
\end{align}
\item[] {\bf Case-VI:} $o_1 = 6$. Here, $o_2 = \gamma\oplus_3\Pi^{0}_{3}(i_2\oplus_2 \delta_{\gamma,2}) ,~~\tilde{\gamma} = \left(o_2 \oplus_3 \hat{\Pi_{3}}(\Pi^{0}_{3}(b))\right),~~ p(y =0| \gamma, o_1=6) =0, $\\$ p(y =1| \gamma, o_1=6) =0~\text{and}~,p(y =2| \gamma, o_1=6) =1$, implying
\begin{align}
&p(\tilde{\gamma}|\gamma, o_1=6)\delta_{\gamma, \tilde{\gamma}_{\oplus_{2}}} 
= \sum_{i_2, o_2}p(o_2|\gamma,\tilde{\gamma}, i_2, o_1=6, y=1 )p(\tilde{\gamma}, i_2 |\gamma,o_1=6, y=0)\delta_{\gamma, \tilde{\gamma}}\nonumber\\
=& \sum_{i_2, o_2} \delta_{o_2, \left(\gamma\oplus_3\Pi^{0}_{3}(i_2\oplus_2 \delta_{\gamma,2})\right)}\delta_{\gamma, \left(o_2 \oplus_3 \hat{\Pi_{3}}(\Pi^{0}_{3}(b))\right)}p(o_2 \oplus_3 \hat{\Pi_{3}}(\Pi^{0}_{3}(b)), i_2 |\gamma,o_1=6, y=1)\nonumber\\\nonumber
=& \sum_{i_2} \delta_{\gamma, \left(\gamma\oplus_3\Pi^{0}_{3}(i_2\oplus_2 \delta_{\gamma,2}) \oplus_3 \hat{\Pi_{3}}(\Pi^{0}_{3}(b))\right)}p(\gamma\oplus_3\Pi^{0}_{3}(i_2\oplus_2 \delta_{\gamma,2}) \oplus_3 \hat{\Pi_{3}}(\Pi^{0}_{3}(b)), i_2 |\gamma,o_1=6, y=1)\\\nonumber
=& \sum_{i_2}\delta_{\gamma, \left(\gamma\oplus_3\Pi^{0}_{3}(i_2\oplus_2 \delta_{\gamma,2}) \oplus_3 \hat{\Pi_{3}}(\Pi^{0}_{3}(b))\right)}p(b, i_2 |\gamma,o_1=6, y=1)\\
=& \sum_{i_2}\delta_{\gamma, \left(\gamma\oplus_3\Pi^{0}_{3}(i_2\oplus_2 \delta_{\gamma,2}) \oplus_3 \hat{\Pi_{3}}(\Pi^{0}_{3}(b))\right)}p(b, i_2 |\gamma, y=1)\nonumber\\
=&\sum_{a}\delta_{\gamma, \left(\gamma\oplus_3\Pi^{0}_{3}(a\oplus_2 \delta_{\gamma,2}) \oplus_3 \hat{\Pi_{3}}(\Pi^{0}_{3}(b))\right)}p(a, b |\gamma, y=1).\label{d6}
\end{align}
\item[] {\bf Case-VII:} $o_1 = 7$. Here, $o_2 = \gamma\oplus_3\Pi^{1}_{3}(i_2\oplus_2 \delta_{\gamma,2}) ,~~\tilde{\gamma} = \left(o_2 \oplus_3 \hat{\Pi_{3}}(\Pi^{1}_{3}(b))\right),~~ p(y =0| \gamma, o_1=7) =0, $\\$ p(y =1| \gamma, o_1=7) =0~\text{and}~,p(y =2| \gamma, o_1=7) =1$  implying
\begin{align}
&p(\tilde{\gamma}|\gamma, o_1=7)\delta_{\gamma, \tilde{\gamma}_{\oplus_{2}}} 
= \sum_{i_2, o_2}p(o_2|\gamma,\tilde{\gamma}, i_2, o_1=7, y=0 )p(\tilde{\gamma}, i_2 |\gamma,o_1=7, y=0)\delta_{\gamma, \tilde{\gamma}}\nonumber\\
=& \sum_{i_2, o_2} \delta_{o_2, \left(\gamma\oplus_3\Pi^{1}_{3}(i_2\oplus_2 \delta_{\gamma,2})\right)}\delta_{\gamma, \left(o_2 \oplus_3 \hat{\Pi_{3}}(\Pi^{1}_{3}(b))\right)}p(o_2 \oplus_3 \hat{\Pi_{3}}(\Pi^{1}_{3}(b)), i_2 |\gamma,o_1=7, y=1)\nonumber\\\nonumber
=& \sum_{i_2} \delta_{\gamma, \left(\gamma\oplus_3\Pi^{1}_{3}(i_2\oplus_2 \delta_{\gamma,2}) \oplus_3 \hat{\Pi_{3}}(\Pi^{1}_{3}(b))\right)}p(\gamma\oplus_3\Pi^{1}_{3}(i_2\oplus_2 \delta_{\gamma,2}) \oplus_3 \hat{\Pi_{3}}(\Pi^{1}_{3}(b)), i_2 |\gamma,o_1=7, y=1)\\\nonumber
=& \sum_{i_2}\delta_{\gamma, \left(\gamma\oplus_3\Pi^{1}_{3}(i_2\oplus_2 \delta_{\gamma,2}) \oplus_3 \hat{\Pi_{3}}(\Pi^{1}_{3}(b))\right)}p(b, i_2 |\gamma,o_1=7, y=1)\\
=& \sum_{i_2}\delta_{\gamma, \left(\gamma\oplus_3\Pi^{1}_{3}(i_2\oplus_2 \delta_{\gamma,2}) \oplus_3 \hat{\Pi_{3}}(\Pi^{1}_{3}(b))\right)}p(b, i_2 |\gamma, y=1)\nonumber\\
=&\sum_{a}\delta_{\gamma, \left(\gamma\oplus_3\Pi^{1}_{3}(a\oplus_2 \delta_{\gamma,2}) \oplus_3 \hat{\Pi_{3}}(\Pi^{1}_{3}(b))\right)}p(a, b |\gamma, y=1).\label{d7}
\end{align}
\end{itemize}
Substituting Eqs.(\ref{d1}-\ref{d7}) in Eq.(\ref{ns3}), for the correlation in Eq.(\ref{ent}), we obtain
\begin{align}
&\$_{\psi^-}(\mathcal{M}_3)
= \frac{1}{21}\sum_{\gamma}\left[p(\tilde{\gamma}|\gamma, o_1=1)\delta_{\gamma, \tilde{\gamma}} + 
                                  p(\tilde{\gamma}|\gamma, o_1=2)\delta_{\gamma, \tilde{\gamma}} +
                                  p(\tilde{\gamma}|\gamma, o_1=3)\delta_{\gamma, \tilde{\gamma}}  \right.\nonumber\\
                                  &\left.\hspace{1cm}+
                                  p(\tilde{\gamma}|\gamma, o_1=4)\delta_{\gamma, \tilde{\gamma}}+
                                  p(\tilde{\gamma}|\gamma, o_1=5)\delta_{\gamma, \tilde{\gamma}} +
                                  p(\tilde{\gamma}|\gamma, o_1=6)\delta_{\gamma, \tilde{\gamma}} +
                                  p(\tilde{\gamma}|\gamma, o_1=7)\delta_{\gamma, \tilde{\gamma}}\right]\nonumber\\ 
=& \frac{1}{21}\left[3 +
\sum_{\gamma, a}
\delta_{\gamma, \left(\hat{\Pi}_3\left(\Pi^{0}_{3}(b)\right)\oplus_3 \left(\gamma\oplus_3\Pi^{0}_{3}(a)\right)\right)}p(a, b |\gamma, y=0)+
\delta_{\gamma, (\hat{\Pi}_3(\Pi^{1}_{3}(b))\oplus_3 \gamma\oplus_3\Pi^{1}_{3}(a))}p(a, b |\gamma, y=0)\right.\nonumber\\
&\left.\hspace{.2cm}+\delta_{\gamma, \left(\gamma\oplus_3\Pi^{0}_{3}(a\oplus_2 \delta_{\gamma,1}) \oplus_3 \hat{\Pi_{3}}(\Pi^{0}_{3}(b))\right)}p(a, b |\gamma, y=1)+
\delta_{\gamma, \left(\gamma\oplus_3\Pi^{1}_{3}(a\oplus_2 \delta_{\gamma,1}) \oplus_3 \hat{\Pi_{3}}(\Pi^{1}_{3}(b))\right)}p(a, b |\gamma, y=1)\right.\nonumber\\
&\left.\hspace{.2cm}+ \delta_{\gamma, \left(\gamma\oplus_3\Pi^{0}_{3}(a\oplus_2 \delta_{\gamma,2}) \oplus_3 \hat{\Pi_{3}}(\Pi^{0}_{3}(b))\right)}p(a, b |\gamma, y=1)+
\delta_{\gamma, \left(\gamma\oplus_3\Pi^{1}_{3}(a\oplus_2 \delta_{\gamma,2}) \oplus_3 \hat{\Pi_{3}}(\Pi^{1}_{3}(b))\right)}p(a, b |\gamma, y=1)
\right]\nonumber\\
=& \frac{1}{7}\left[1+\frac{3}{8}\times8
                     +\frac{1}{2}\times4\right]=\frac{6}{7}~~.
\end{align}
\textbf{Classical Optimal Success Probability}\\
Here, we show the optimal success probability of sending three symbols $\gamma\in \{0,1,2\}$ through $\mathcal{M}_3$ using classical resources. The inputs $\{(i_1, i_2)\}$ can be divided into two disjoint sets, $\mathcal{T}_1 = \{(0,0), (1,1), (2,1)\}$ and $\mathcal{T}_2 = \{(0,1), (1,0), (2,0)\}$ where, each pair of inputs belonging to the set $\mathcal{T}_i$ are confusable by two outputs, for all $i$. There are two choices for encoding the three symbols. Either, all three symbols are encoded to the inputs belonging to same set $\mathcal{T}_i$ or any two are encoded in $\mathcal{T}_i$ and the remaining one in the set $\mathcal{T}_{i'}$ where $i' \neq i$. Let Bob's guess be $\tilde{\gamma}$, then the success probability is:\\
\begin{align}
\$_{Cl}(\mathcal{M}_3)&=\sum_{\gamma}p(\tilde{\gamma} = \gamma ,\gamma)
=\sum_{\gamma} p(\gamma) p(\tilde{\gamma} |\gamma)\delta_{\gamma, \tilde{\gamma}}=\frac{1}{3}\sum_{\gamma} p(\tilde{\gamma} |\gamma)\delta_{\gamma, \tilde{\gamma}}\nonumber\\
&=\frac{1}{3}\sum_{\gamma,o_1,o_2} p(\tilde{\gamma},o_1,o_2 |\gamma)\delta_{\gamma, \tilde{\gamma}}\nonumber\\
&=\frac{1}{3}\sum_{\gamma,o_1,o_2} p(\tilde{\gamma} |\gamma,o_1,o_2)p(o_1,o_2 |\gamma)\delta_{\gamma, \tilde{\gamma}}~.
\end{align}
If for an input pair, there are two confusing outcomes then any deterministic decoding of outcomes will consistently prefer one input (symbol) over the other, leading to incorrect decoding for at least one of the inputs for both outcomes. Consider the following two encoding cases:
\begin{itemize}
\item {\bf Case-I:} The classical messages $\gamma\in\{0,1,2\}$ are encoded in inputs from the same set \(\mathcal{T}_i\). Each pair of inputs will result in two confusing outcomes. This implies that there will be at least six outcomes for which there is always incorrect decoding for at least one of the inputs from each pair. And hence, out of the $21$ possible outcomes corresponding to the three inputs, $9$ are unique and not confusing while the remaining $12$ are not all unique but rather a repetition of $6$ unique ones. Then the success probability maximized over all decodings will be $\frac{15}{21}$ for this encoding.

\item {\bf Case-II:} Two of the messages are encoded in inputs from $\mathcal{T}_i$ and the other one in an input from the set $\mathcal{T}_{i'}$. The pair belonging to $\mathcal{T}_i$ will be confused by two outcomes and the other two pairs with a single outcome. This implies that there will be at least four outcomes for which the decoding is incorrect for at least one of the inputs from each pair. Hence, out of the $21$ possible outcomes corresponding to the three inputs, $13$ unique outcomes are not confusing, which leads to perfect decoding. The remaining $8$ outcomes are not all unique, but rather a repetition of $4$ unique ones leading to confusion among the pairs of inputs. Then the success probability maximized over all decodings will be $\frac{17}{21}$ for this encoding.

\begin{align}
\$^{max}_{Cl}(\mathcal{M}_3)&=  \frac{17}{21} < \frac{18}{21} =\frac{6}{7}~.
\end{align}
\end{itemize}


\begin{thebibliography}{65}
\bibitem [{Sel()}]{Self1}%
  \BibitemOpen
  \bibfield  {title} {\bibinfo {title} {Given two graphs $\mathbb{G}_1$ and $\mathbb{G}_2$ with respective vertex sets $\mathbb{V}_1$ \& $\mathbb{V}_2$ and edge sets $\mathbb{E}_1$ \& $\mathbb{E}_2$, their strong product $\mathbb{G}_1\boxtimes \mathbb{G}_2$ is a graph on $\mathbb{V}_1\times\mathbb{V}_2$ with an edge $\{(v_1,v_2),(v^\prime_1,v^\prime_2)\}$ iff $[(\{v_1,v^\prime_1\}\in\mathbb{E}_1)\land(\{v_2,v^\prime_2\}\in\mathbb{E}_2)]\lor[(v_1=v^\prime_1)\land(\{v_2,v^\prime_2\}\in\mathbb{E}_2)]\lor[(\{v_1,v^\prime_1\}\in\mathbb{E}_1)\land (v_2=v^\prime_2)]$},\ }\href@noop {} {\ }\BibitemShut {NoStop}%
\end{thebibliography}
%

\end{document}